\documentclass{llncs}

\usepackage{microtype}
\usepackage{tensor}
\usepackage{amsmath,boxedminipage}
\usepackage{amssymb}
\usepackage{tikz}
\usepackage{enumerate}
\usepackage{hyperref}
\usepackage[ruled]{algorithm2e}
\usepackage{comment}
\usetikzlibrary{arrows,shapes,calc,positioning}
\newtheorem{open}{Open Problem}

\title{Recognizing Graphs Close to Bipartite Graphs\\with an Application to Colouring Reconfiguration\thanks{This paper received support from EPSRC (EP/K025090/1), London Mathematical Society (41536), the Leverhulme Trust (RPG-2016-258) and Fondation Sciences Math\'ematiques de Paris.
An extended abstract of this paper appeared in the proceedings of MFCS 2017~\cite{BDFJP17-conf}.}}

\author{Marthe Bonamy\inst{1} \and Konrad K. Dabrowski\inst{2} \and Carl Feghali\inst{3} \and\\ Matthew Johnson\inst{2} \and Dani\"el Paulusma\inst{2}}
\institute{
CNRS, LaBRI, Universit\'e de Bordeaux, France\\
\texttt{marthe.bonamy@u-bordeaux.fr},
\and
School of Engineering and Computing Sciences, Durham University, UK\\
\texttt{\{konrad.dabrowski,matthew.johnson2,daniel.paulusma\}@durham.ac.uk},
\and
IRIF \& Universit\'e Paris Diderot, France\\
\texttt{feghali@irif.fr}}

\newcommand{\NP}{{\sf NP}}

\newcommand{\problemdef}[3]{
	\begin{center}
		\begin{boxedminipage}{.99\textwidth}
			\textsc{{#1}}\\[2pt]
			\begin{tabular}{ r p{0.8\textwidth}}
				\textit{~~~~Instance:} & {#2}\\
				\textit{Question:} & {#3}
			\end{tabular}
		\end{boxedminipage}
	\end{center}
}

\newcounter{ctrclaim}[theorem]

\newcommand\displaycase[1]{{\em #1}}
\newcommand{\clm}[1]{\medskip\phantomsection\refstepcounter{ctrclaim}\noindent\displaycase{Claim \thectrclaim. }{\em #1}\\}

\newcommand{\qedllncs}{\qed}

\begin{document}
\maketitle

\begin{abstract}
We continue research into a well-studied family of problems that ask whether the vertices of a graph can be partitioned into sets~$A$ and~$B$, where~$A$ is an independent set and~$B$ induces a graph from some specified graph class~${\cal G}$.
We let~${\cal G}$ be the class of $k$-degenerate graphs.
This problem is known to be polynomial-time solvable if $k=0$ (bipartite graphs) and \NP-complete if $k=1$ (near-bipartite graphs) even for graphs of maximum degree~$4$.
Yang and Yuan [DM, 2006] showed that the $k=1$ case is polynomial-time solvable for graphs of maximum degree~$3$.
This also follows from a result of Catlin and Lai [DM, 1995].
We consider graphs of maximum degree~$k+2$ on~$n$ vertices.
We show how to find~$A$ and~$B$ in~$O(n)$ time for $k=1$, and in~$O(n^2)$ time for $k\geq 2$.
Together, these results provide an algorithmic version of a result of Catlin [JCTB, 1979] and also provide an algorithmic version of a generalization of Brook's Theorem, which was proven in a more general way by Borodin, Kostochka and Toft [DM, 2000] and Matamala [JGT, 2007].
Moreover, the two results enable us to complete the complexity classification of an open problem of Feghali et~al.~[JGT, 2016]: finding a path in the vertex colouring reconfiguration graph between two given $\ell$-colourings of a graph of maximum degree~$k$.

\end{abstract}

\section{Introduction}\label{s-intro}

The {\sc Colouring} problem asks if a given graph is {\em $k$-colourable} for some given integer~$k$, that is, if the vertices of the graph can be coloured with at most~$k$ colours, such that no two adjacent vertices are coloured alike.
This is a central problem in graph theory and well known to be \NP-complete even if $k=3$~\cite{Lo72}.
A stronger property of a graph is that of being {\em $(k-1)$-degenerate}, which is the case when every induced subgraph has a vertex of degree at most~$k-1$: every $(k-1)$-degenerate graph is $k$-colourable, but the converse is not true.

For an arbitrarily large integer~$k$, there exist $k$-degenerate graphs that are not $(k-1)$-degenerate but that can be decomposed into a $p$-degenerate induced subgraph and a $q$-degenerate induced subgraph for two small integers~$p$ and~$q$.
For example, if we take complete bipartite graphs of degree~$k$, then we can let $p=q=0$.
If a $k$-degenerate graph is decomposable in this way, it is not only $(k+1)$-colourable but even $(p+q+2)$-colourable.
This leads to the well-studied problem of identifying graphs whose vertex sets can be partitioned into two sets~$A$ and~$B$ such that~$A$ and~$B$ induce a $p$-degenerate graph and $q$-degenerate graph, respectively; see, for instance,~\cite{Bo76,BG01,KT09,Ma07,Th95,Th01}.
If a graph has such a partition with $p+q=\ell-2$, then we can say that the graph is ``robustly'' $\ell$-colourable.
For the sake of example: every planar graph is $5$-degenerate, which implies that it is $6$-colourable, but we can improve this to $5$-colourable by applying the result of Thomassen~\cite{Th01}, which states that every planar graph can be decomposed into a $0$-degenerate graph and a $3$-degenerate graph, or another result of Thomassen~\cite{Th95}, which states that every planar graph can be decomposed into a $1$-degenerate graph and a $2$-degenerate graph.
So being $5$-colourable is a ``robust'' property of planar graphs (in contrast to being $4$-colourable; there is no~$p$ and~$q$ with $p+q\leq 2$ such that we can guarantee a decomposition of a planar graph into a $p$-degenerate graph and a $q$-degenerate graph; see~\cite{HSW90} for $p=0$, $q=2$ and~\cite{CK69} for $p=q=1$).

\medskip
\noindent
{\bf Research Question.}
We will apply the notion of ``robust'' $k$-colourability to a central problem in the area of graph reconfiguration, that of finding a path between two given $k$-colourings in the $k$-colouring reconfiguration graph~$R_k(G)$ of a graph~$G$.
The graph~$R_k(G)$ has as vertices the $k$-colourings of~$G$ and two $k$-colourings are adjacent if and only if they differ on exactly one vertex of~$G$.
This problem is PSPACE-hard even if $k=4$ and~$G$ is planar bipartite~\cite{BC09}.
In its complexity classification for graphs of maximum degree~$\Delta$~\cite{FJP16} there is one open case $(k,\Delta)$ left.
As argued in~\cite{FJP16}, in order to solve this case we must answer the following research question:

\medskip
\noindent
{\em Is it possible to find in polynomial time a partition $(A,B)$ of the vertex set of a graph~$G$ of maximum degree~$k$, such that~$A$ is $0$-degenerate and~$B$ induces a $(k-2)$-degenerate graph?}

\subsection{Known Existence Results}
We note that in the above question we must {\em find} a partition $(A,B)$ for $p=0$ and $q=k-2$.
This is a different question than deciding whether such a partition {\em exists}.
For the latter question, a number of different results exist in the literature.
We will survey these results below, as they are very insightful for our question, although they do not solve it.

We first note that if $p=0$, then we can take~$q$ as a distance measure to control how ``far'' the graph is from being bipartite.
As every $0$-degenerate graph is an independent set of vertices, checking if the distance is~$0$ is the same as checking bipartiteness, which can be solved in linear time.
Graphs within distance~$1$ from being bipartite are said to be {\em near-bipartite}.
By definition, such a graph has a {\em near-bipartite decomposition}, that is, a partition $(A,B)$ of its vertex set, where~$A$ is an independent set and~$B$ induces a $1$-degenerate graph, or equivalently, a forest.
Deciding whether a graph is near-bipartite is \NP-complete~\cite{BLS98}.

Yang and Yuan~\cite{YY06} proved that the problem remains \NP-complete even for graphs of maximum degree~$4$, but becomes polynomial-time solvable for graphs of maximum degree~$3$.
To prove the latter result, they showed that every connected graph of maximum degree at most~$3$ is near-bipartite except~$K_4$ (we let~$K_k$ denote the complete graph on~$k$ vertices).
This characterization was also shown by Catlin and Lai, who proved that the independent set~$A$ may even be assumed to be maximum.

\begin{theorem}[\cite{CL95}]\label{t-cl95}
The vertex set of every connected graph of maximum degree~$3$ that is not isomorphic to~$K_4$ can be partitioned into two sets~$A$ and~$B$, where~$A$ is a maximum independent set and~$B$ is a forest.
\end{theorem}
Theorem~\ref{t-cl95} generalizes the $k=3$ case of an earlier result of Catlin.

\begin{theorem}[\cite{Catlin79}]\label{t-ca79}
For every integer $k\geq 3$, the vertex set of every connected graph of maximum degree~$k$ that is not isomorphic to~$K_{k+1}$ can be partitioned into two sets~$A$ and~$B$, where~$A$ is a maximum independent set and~$B$ induces a graph that does not contain a~$K_k$.
\end{theorem}
Matamala generalized both Theorem~\ref{t-cl95} and Theorem~\ref{t-ca79}.

\begin{theorem}[\cite{Ma07}]\label{t-ma07}
For every three integers $k\geq 3$ and $p,q\geq 0$ with $p+q=k-2$, the vertex set of every connected graph of maximum degree~$k$ that is not isomorphic to~$K_{k+1}$ can be partitioned into two sets~$A$ and~$B$, where~$A$ induces a $p$-degenerate subgraph of maximum size and~$B$ induces a $q$-degenerate subgraph.
\end{theorem}
Before Theorem~\ref{t-ma07} appeared, Borodin, Kostochka and Toft proved a more general result, except that the property of the first set having maximum size is not assumed.
We present a simpler version of their result and refer to~\cite{BKT00} for full details.

\begin{theorem}[\cite{BKT00}]\label{t-bkt00}
For all integers $k\geq 3$, $s \geq 2$, and $p_1,\ldots,p_s\geq 0$ such that $p_1+\cdots+p_s=k-s$, the vertex set of every connected graph of maximum degree~$k$ that is not isomorphic to~$K_{k+1}$ can be partitioned into sets~$A_1,\ldots,A_s$, where~$A_i$ induces a $p_i$-degenerate subgraph for $i\in\{1,\ldots,s\}$.
\end{theorem}

Brooks' Theorem~\cite{Br41} states that every graph~$G$ with maximum degree~$k\geq 2$ is $k$-colourable unless~$G$ is a complete graph or an odd cycle.
Recall that every $(k-1)$-degenerate graph is $k$-colourable.
Hence Theorems~\ref{t-ma07} and~\ref{t-bkt00}, together with the trivial case $k=2$, generalize Brooks' Theorem.

By choosing $p=0$ in Theorem~\ref{t-ma07} we obtain the following special case, which implies that every connected graph of maximum degree~$k$ except~$K_{k+1}$ is within distance~$k-2$ from being bipartite.

\begin{theorem}\label{t-special}
For every integer $k\geq 3$, the vertex set of every connected graph of maximum degree~$k$ that is not isomorphic to~$K_{k+1}$ can be partitioned into two sets~$A$ and~$B$, where~$A$ is a maximum independent set and~$B$ induces a $(k-2)$-degenerate subgraph.
\end{theorem}
Theorem~\ref{t-special} only guarantees that the existence of a desired partition $(A,B)$ can be tested in polynomial time.
It does not tell us how to find such a partition.
Obtaining an algorithmic version of Theorem~\ref{t-special} corresponds to our research question.
We cannot hope to keep the condition that~$A$ is a maximum independent set, as this would require solving an \NP-complete problem: {\sc Independent Set} for connected cubic graphs~\cite{GJS76}.
Before giving our results we first survey some other algorithmic results.

\subsection{Known Algorithmic Results}

\noindent
{\bf Special Graph Classes.}
As discussed, Yang and Yuan~\cite{YY06} proved that recognizing near-bipartite graphs is polynomial-time solvable for graphs of maximum degree~$k$ when $k\leq 3$ and \NP-complete when $k\geq 4$.
They also proved that recognizing near-bipartite graphs of diameter~$k$ is polynomial-time solvable when $k\leq 2$ and \NP-complete when $k\geq 4$.
Recently, we solved their missing case by proving that the problem is \NP-complete for graphs of diameter~$3$~\cite{BDFJP17c}.
Brandst\"adt et~al.~\cite{BBKNP13} proved that recognizing near-bipartite perfect graphs is \NP-complete.
We also proved that recognizing near-bipartite graphs is \NP-complete for line graphs of maximum degree~$4$~\cite{BDFJP17b}.

Borodin and Glebov~\cite{BG01} showed that every planar graph of girth at least~$5$ is near-bipartite (see~\cite{KT09} for an extension of this result).
Dross, Montassier and Pinlou~\cite{DMP16} asked whether every triangle-free planar graph is near-bipartite.
In fact they proved that if this is not the case, then the problem of recognizing near-bipartite graphs is \NP-complete for triangle-free planar graphs.
Their construction can be easily modified to prove that the problem of recognizing near-bipartite graphs is \NP-complete for planar graphs.\footnote{The \NP-hardness reduction in~\cite{DMP16} uses a minimal triangle-free planar graph~$G$ that is not near-bipartite, and it is not known whether such graphs exist.
If we remove the triangle-free condition, we can replace~$G$ by~$K_4$.}

\medskip
\noindent
{\bf Minimum Independent Feedback Vertex Sets.}
The problem of finding a decomposition into an independent set~$A$ and a forest~$B$ where the size of~$A$ is minimum has also been studied.
In this context~$A$ is said to be an {\em independent feedback vertex set}.
Computing a minimum independent feedback vertex set has been shown to be \NP-hard even for planar bipartite graphs of maximum degree~$4$, but linear-time solvable for graphs of bounded treewidth, chordal graphs and $P_4$-free graphs~\cite{TIZ15}
(it was already known that near-bipartite $P_4$-free graphs can be recognized in linear time~\cite{BBKNP13}).
Recently, we extended the polynomial-time result from~\cite{TIZ15} for $P_4$-free graphs to $P_5$-free graphs~\cite{BDFJP17b}.
We also gave a polynomial-time algorithm for computing a minimum independent feedback vertex set of a graph of diameter~$2$~\cite{BDFJP17c}.

The problem of computing small independent feedback vertex sets has also been studied from the perspective of parameterized complexity.
In this setting, the size of~$A$ is taken as the parameter.
Misra et~al.~\cite{MPRS12} gave the first FPT algorithm for this problem, which was later improved by Agrawal et~al.~\cite{AGSS16}.

\medskip
\noindent
{\bf Problem Variants.}
The {\sc Induced Forest $2$-Partition} problem is closely related to the problem of recognizing near-bipartite graphs.
It asks whether the vertex set of a given graph can be decomposed into two disjoint sets~$A$ and~$B$, where both~$A$ and~$B$ induce forests.
Wu, Yuan and Zhao~\cite{WYZ96} proved that {\sc Induced Forest $2$-Partition} is \NP-complete for graphs of maximum degree~$5$ and polynomial-time solvable for graphs of maximum degree at most~$4$.
The problem variant where the maximum degree of one of the two induced forests is bounded by some constant has also been studied, in particular from a structural point of view (see, for instance,~\cite{DMP16}).

A similar problem, known as {\sc Dominating Induced Matching}, asks whether the vertex set of a graph can be partitioned into an independent set and an induced matching (a set of isolated edges) and was shown by Grinstead et~al.~\cite{GSSH93} to be \NP-complete.
Brandst{\"a}dt, Le and Szymczak~\cite{BLS98} proved \NP-completeness of another closely related problem, that of deciding whether the vertex set of a given graph can be decomposed into an independent set and a tree.
As trees, induced matchings and forests are $2$-colourable, these two problems and that of recognizing near-bipartite graphs can be seen as restricted variants of the {\sc $3$-Colouring} problem.
This problem is well known to be \NP-complete~\cite{Lo73}.
However, the \NP-hardness result of Brandst\"adt~et~al.~\cite{BBKNP13} for perfect graphs shows that there are hereditary graph classes on which the complexities of recognizing near-bipartite graphs and {\sc $3$-Colouring} do not coincide, as {\sc $3$-Colouring} (or even {\sc $k$-Colouring} with~$k$ part of the input~\cite{GLS84}) is polynomial-time solvable for perfect graphs.

A $3$-colouring of a graph is acyclic if {\em every} two colour classes induce a forest.
We observe that every graph with an acyclic $3$-colouring is near-bipartite, but the reverse is not necessarily true.
Kostochka~\cite{Ko78} proved that the corresponding decision problem {\sc Acyclic $3$-Colouring} is \NP-complete.
Later, Ochem~\cite{Oc05} showed that {\sc Acyclic $3$-Colouring} is \NP-complete even for planar bipartite graphs of maximum degree~$4$.
As every bipartite graph is near-bipartite, this result implies that there are hereditary graph classes on which the complexities of recognizing near-bipartite graphs and {\sc Acyclic $3$-Colouring} do not coincide.

\medskip
\noindent
{\bf Generalizations.}
For a fixed graph class~${\cal G}$ (that is,~${\cal G}$ is not part of the input), Brandst{\"a}dt, Le and Szymczak~\cite{BLS98} considered the following more general problem:

\problemdef{\sc Stable(${\cal G}$)}
{A graph $G=(V,E)$.}{Can~$V$ be decomposed into two disjoint sets~$A$ and~$B$, where~$A$ is an independent set and~$B$ induces a graph in~${\cal G}$?}
Note that {\sc Stable(${\cal G}$)} is equivalent to {\sc $3$-Colouring} if we choose~${\cal G}$ to be the class of bipartite graphs.
If we choose~${\cal G}$ to be the class of $(k-2)$-degenerate graphs, then we obtain the decision version of the problem we consider in this paper.

If~${\cal G}$ is the class of complete graphs, then {\sc Stable(${\cal G}$)} is the problem of recognizing split graphs, and this problem can be solved in polynomial time.
Brandst{\"a}dt, Le and Szymczak~\cite{BLS98} proved that {\sc Stable(${\cal G}$)} is \NP-complete when~${\cal G}$ is the class of trees or the class of trivially perfect graphs, and polynomial-time solvable when~${\cal G}$ is the class of co-bipartite graphs, the class of split graphs, or the class of threshold graphs.
Moreover, {\sc Stable(${\cal G}$)} has also been shown to be \NP-complete when~${\cal G}$ is the class of triangle-free graphs~\cite{CC96}, the class of $P_4$-free graphs~\cite{HL00}, the class of graphs of maximum degree~$1$~\cite{MP95}, or, more generally, a class of graphs that has any additive hereditary property not equal to or divisible by the property of being edgeless~\cite{KS97}, whereas it is also polynomial-time solvable if~${\cal G}$ is the class of complete bipartite graphs~\cite{FHKM03} (see~\cite{BHLL05} for a faster algorithm).
The {\sc Stable(${\cal G}$)} problem has also been studied for hereditary graph classes~${\cal G}$ with subfactorial or factorial speed~\cite{DLS15,Lo05} (the speed of a graph class is the function that given an integer~$n$ returns the number of labelled graphs on~$n$ vertices in the class).

By relaxing the condition on the set~$A$ being independent we obtain the more general problem of $({\cal G}_1,{\cal G}_2)$-{\sc Recognition}, which asks whether the vertex set of a graph can be decomposed into disjoint sets~$A$ and~$B$, such that~$A$ induces a graph in~${\cal G}_1$ and~$B$ induces a graph in~${\cal G}_2$.
For instance, if~${\cal G}_1$ is the class of cliques and~${\cal G}_2$ is the class of disjoint unions of cliques, then the $({\cal G}_1,{\cal G}_2)$-{\sc Recognition} problem is equivalent to recognizing unipolar graphs (see~\cite{MY15} for a quadratic algorithm).
Generalizing {\sc Stable(${\cal G}$)} can also lead to a family of transversal problems, such as {\sc Feedback Vertex Set}.
However, such generalizations are beyond the scope of our paper.

\subsection{Our Results}

In Section~\ref{s-3}, we consider near-bipartite decompositions of \emph{subcubic} graphs (that is, graphs of maximum degree at most~$3$).
Recall that by Theorem~\ref{t-cl95} the only connected subcubic graph that is not near-bipartite is~$K_4$ (see also~\cite{YY06}) and so near-bipartite subcubic graphs can be recognized in polynomial time.
However, neither the proof of Theorem~\ref{t-cl95} in~\cite{CL95} nor the proof of this fact in~\cite{YY06} leads to a linear-time algorithm for finding the desired partition $(A,B)$.
As mentioned, in the case of Theorem~\ref{t-cl95}, this would in fact require solving an \NP-complete problem: {\sc Independent Set} for cubic graphs~\cite{GJS76}.
We give an $O(n)$-time algorithm that finds a near-bipartite decomposition of any subcubic graph with no component isomorphic to~$K_4$.

We say a partition $(A,B)$ of the vertex set of a graph is a {\em $k$-degenerate decomposition} if~$A$ is independent and~$B$ induces a $(k-2)$-degenerate graph, so Section~\ref{s-3} is concerned with $3$-degenerate decompositions of graphs of maximum degree~$3$.
In Section~\ref{s-maxdegree}, we consider, more generally, $k$-degenerate decompositions of graphs of maximum degree at most~$k$ for any $k \geq 3$.
By Theorem~\ref{t-special}, the only connected graph with maximum degree~$k$ that does not have a $k$-degenerate decomposition is~$K_{k+1}$.
As mentioned, Theorem~\ref{t-special} does not imply a polynomial-time algorithm for finding such a decomposition.
We give an $O(n^2)$-time algorithm to find such a decomposition for any $k \geq 3$ (in contrast with the $O(n)$-time algorithm in Section~\ref{s-3} for the special case when $k=3$).

Our results in Sections~\ref{s-3} and~\ref{s-maxdegree} provide an algorithmic version of Theorem~\ref{t-special} and, as Theorem~\ref{t-special} generalizes Theorem~\ref{t-ca79}, they also imply an algorithmic version of Theorem~\ref{t-ca79}.

In Section~\ref{s-hard} we prove that the problem of deciding whether a graph of maximum degree~$2k-2$ has a $k$-degenerate decomposition is \NP-complete.
We do this by adapting the proof of the aforementioned result of Yang and Yuan~\cite{YY06}, which states that recognizing near-bipartite graphs of maximum degree~$4$ is \NP-complete (the $k=3$ case).
In Section~\ref{s-reconfig} we apply our algorithms from Sections~\ref{s-3} and~\ref{s-maxdegree} to completely settle the complexity classification of the graph colouring reconfiguration problem considered in~\cite{FJP16}.
Finally, in Section~\ref{s-con}, we give directions for future work.

\section{Linear Time for Graphs of Maximum Degree at Most~$3$}\label{s-3}

In this section we prove the following result.

\begin{theorem}\label{thm:cubic}
Let~$G$ be a subcubic graph on~$n$ vertices, with no component isomorphic to~$K_4$.
Then a near-bipartite decomposition of~$G$ can be found in~$O(n)$ time.
\end{theorem}

\begin{proof}
We will repeatedly apply a set of \emph{rules} to~$G$.
Each rule takes constant time to apply and after each application of a rule, the resulting graph contains fewer vertices.
The rules are applied until the empty graph is obtained.
We then reconstruct~$G$ from the empty graph by working through the rules applied in reverse order.
As we rebuild~$G$ in this way, we find a near-bipartite decomposition of each obtained graph.
We do this by describing how to extend, in constant time, a near-bipartite decomposition of a graph before some rule is undone to a near-bipartite decomposition of the resulting graph after that rule is undone.
If we can do this then we say that the rule is {\em safe}.
We conclude that the total running time of the algorithm is~$O(n)$.
It only remains to describe the rules, show that it takes constant time to do and undo each of them and prove that they are safe.

We need the following terminology.
The {\em claw} is the graph with vertices $u,v_1,v_2,v_3$ and edges~$uv_1$, $uv_2$,~$uv_3$; the vertex~$u$ is the {\em centre} of the claw (see \figurename~\ref{fig:prism-claw}).
The {\em triangular prism} is the graph obtained from two triangles on vertices $u_1,u_2,u_3$ and $v_1,v_2,v_3$, respectively, by adding the edges~$u_iv_i$ for $i \in \{1,2,3\}$ (see \figurename~\ref{fig:prism-claw}).
Two vertices are {\em false twins} if they have the same neighbourhood (note that such vertices must be non-adjacent).

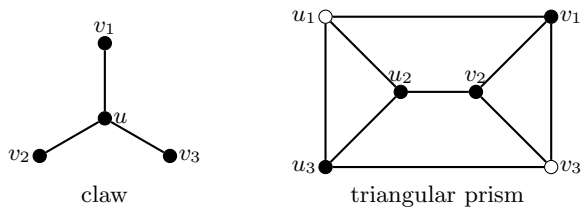
\begin{figure}
\tikzstyle{vertex}=[circle,draw=black, fill=black, minimum size=5pt, inner sep=1pt]
\tikzstyle{edge} =[draw,-,black,>=triangle 90]
\begin{center}
\begin{tabular}{cc}
\begin{minipage}{0.33\textwidth}
\begin{center}
\begin{tikzpicture}
   \foreach \pos/\name / \label / \posn / \dist  in {
{(0,0)/u/u/{right}/0},
{(90:1)/v1/v_1/{above}/0},
{(210:1)/v2/v_2/{left}/0},
{(330:1)/v3/v_3/{right}/0}}
      { \node[vertex] (\name) at \pos {};
       \node [\posn=\dist] at (\name) {$\label$};
       }

\foreach \source/ \dest  in {u/v1, u/v2, u/v3}
       \path[edge, black,  thick] (\source) --  (\dest);     
\end{tikzpicture}
\end{center}
\end{minipage}
&
\begin{minipage}{0.33\textwidth}
\begin{center}
\begin{tikzpicture}
   \foreach \pos/\name / \label / \posn / \dist  in {
{(1,1)/u2/u_2/{above}/0},
{(0,0)/u3/u_3/{left}/0},
{(3,2)/v1/v_1/{right}/0},
{(2,1)/v2/v_2/{above}/0}}
      { \node[vertex] (\name) at \pos {};
       \node [\posn=\dist] at (\name) {$\label$};
       }

   \foreach \pos/\name / \label / \posn / \dist  in {
{(0,2)/u1/u_1/{left}/0},
{(3,0)/v3/v_3/{right}/0}}
      { \node[vertex, fill=white] (\name) at \pos {};
       \node [\posn=\dist] at (\name) {$\label$};
       }

\foreach \source/ \dest  in {u1/u2, u1/u3, u2/u3, v1/v2, v1/v3, v2/v3, u1/v1, u2/v2, u3/v3}
       \path[edge, black,  thick] (\source) --  (\dest);     
\end{tikzpicture}
\end{center}
\end{minipage}\\
claw & triangular prism
\end{tabular}
\end{center}
\caption{The claw and the triangular prism.
A near-bipartite decomposition of the triangular prism is indicated: the white vertices form an independent set and the black vertices induce a forest.}
\label{fig:prism-claw}
\end{figure}

Let~$u$ be an arbitrary vertex of~$G$.
Our choice of~$u$ as an arbitrary vertex implies that~$u$ can be found in constant time.
We then use the first of the following rules that is applicable.

\begin{enumerate}[\bf Rule 1.]
\item \label{r-1}If~$u$ has degree at most~$2$, then remove~$u$.
\item \label{r-2}If there is a vertex~$v$ of degree at most~$2$ that is at distance at most~$3$ from~$u$, then remove~$v$.
\item \label{r-3}If~$G$ contains an induced diamond~$D$ whose vertices are at distance at most~$3$ from~$u$, then remove the vertices of~$D$.
\item \label{r-4}If there is a pair of false twins $u_1$, $u_2$ each at distance at most~$2$ from~$u$, then remove $u_1$, $u_2$ and their common neighbours (note that $u\in \{u_1,u_2\}$ is possible).
\item \label{r-5}If~$u$ is in a connected component that is a triangular prism~$P$, then remove the vertices of~$P$.
\item \label{r-6}If Rules~\ref{r-1}--\ref{r-5} do not apply but~$u$ is in a triangle~$T$, then the neighbours of the vertices in~$T$ that are outside~$T$ are pairwise distinct (since there is no induced diamond) and at least two them, which we denote by~$x'$, $y'$, are non-adjacent (otherwise~$u$ belongs to a triangular prism).
Remove the vertices of~$T$ and add an edge between~$x'$ and~$y'$.
\item \label{r-7}If~$u$ is the centre of an induced claw but has a neighbour~$v$ that belongs to a triangle, then apply one of the Rules~\ref{r-1}--\ref{r-6} on~$v$.
\item \label{r-8}If the graph induced by the vertices at distance at most~$3$ from~$u$ contains the graph~$H_1$, $H_2$ or~$H_3$, depicted in \figurename~\ref{fig:rule8}, with the vertex~$u$ in the position shown in the figure, then remove the vertices of this graph~$H_i$.
\item \label{r-9}If Rules~\ref{r-1}--\ref{r-8} do not apply but~$u$ is the centre of an induced claw and its three neighbours~$u_1$, $u_2$,~$u_3$ are also centres of induced claws, then remove~$u$, $u_1$, $u_2$,~$u_3$ and for $i\in\{1,2,3\}$ add an edge joining the two neighbours of~$u_i$ distinct from~$u$ and denote it by~$e_i$; we say that such an edge is {\em new} (note that such neighbours of two distinct~$u_i$ and~$u_j$ may overlap).
\end{enumerate}

\begin{figure}
\tikzstyle{vertex}=[circle,draw=black, fill=black, minimum size=5pt, inner sep=1pt]
\tikzstyle{edge} =[draw,-,black,>=triangle 90]
\begin{center}
\begin{tabular}{ccc}
\begin{minipage}{0.3\textwidth}
\begin{center}
\begin{tikzpicture}
       
    \begin{scope}[yshift=-1cm]   
   \foreach \pos/\name / \label / \posn / \dist in {   
    {(7,0)/v/u/{right}/0},
    {(6.5,-1)/v1/u_1/{right}/0},
    {(7.5,-1)/v2/u_2/{right}/0}, 
    {(6,-2)/t1/v_1/{right}/0},
     {(7,-3)/t4/w/{right}/0}}
       { \node[vertex] (\name) at \pos {};
       \node [\posn=\dist] at (\name) {$\label$};
       }

   \foreach \pos/\name / \label / \posn / \dist in {   
    {(7,-2)/t2/v_2/{right}/0},
    {(8,-2)/t3/v_3/{right}/0}}
       { \node[vertex, fill=white] (\name) at \pos {};
       \node [\posn=\dist] at (\name) {$\label$};
       }

       \path[edge, black,  thick] (t1) to [bend right=30] (t3);
\foreach \source/ \dest  in {v/v1,  v/v2, v1/t1, v1/t2, v2/t2, v2/t3, t1/t4, t2/t4, t3/t4}
       \path[edge, black,  thick] (\source) --  (\dest);

\end{scope}
\end{tikzpicture}
\end{center}
\end{minipage}
&
\begin{minipage}{0.3\textwidth}
\begin{center}
\begin{tikzpicture}

       \begin{scope}[xshift=3.5cm, yshift=-1cm]   
   \foreach \pos/\name / \label / \posn / \dist in {   
    {(7,0)/v/u/{right}/0},
    {(6.,-1)/v1/u_1/{right}/0},
    {(7,-1)/v2/u_2/{right}/0},
    {(7.5,-2)/t3//{right}/0},
    {(8.5,-2)/t4//{right}/0}}
       { \node[vertex] (\name) at \pos {};
       \node [\posn=\dist] at (\name) {$\label$};
       }

   \foreach \pos/\name / \label / \posn / \dist in {   
    {(8,-1)/v3/u_3/{right}/0},    
    {(5.5,-2)/t1//{right}/0},
    {(6.5,-2)/t2//{right}/0}}
       { \node[vertex, fill=white] (\name) at \pos {};
       \node [\posn=\dist] at (\name) {$\label$};
       }

       \path[edge, black,  thick] (t1) to [bend right=30] (t3);
       \path[edge, black,  thick] (t2) to [bend right=30] (t4);
       \path[edge, black,  thick] (t1) to [bend right=50] (t4);

\foreach \source/ \dest  in {v/v1,  v/v2, v/v3, v1/t1, v1/t2, v2/t2, v2/t3, v3/t3, v3/t4}
       \path[edge, black,  thick] (\source) --  (\dest);

\end{scope}
\end{tikzpicture}
\end{center}
\end{minipage}
&
\begin{minipage}{0.3\textwidth}
\begin{center}
\begin{tikzpicture}
\begin{scope}[xshift=7cm, yshift=-1cm]
   \foreach \pos/\name / \label / \posn / \dist in {   
    {(7,0)/v/u/{right}/0},
    {(6,-2)/t1//{right}/0},
    {(7,-2)/t2//{right}/0},
    {(8,-2)/t3//{right}/0},
    {(7,-3)/t4//{right}/0}}
       { \node[vertex] (\name) at \pos {};
       \node [\posn=\dist] at (\name) {$\label$};
       }

   \foreach \pos/\name / \label / \posn / \dist in {   
    {(6.,-1)/v1/u_1/{right}/0},
    {(7,-1)/v2/u_2/{right}/0},
    {(8,-1)/v3/u_3/{right}/0}}
       { \node[vertex, fill=white] (\name) at \pos {};
       \node [\posn=\dist] at (\name) {$\label$};
       }

\foreach \source/ \dest  in {v/v1,  v/v2, v/v3, v1/t1, v1/t2, v2/t2, v2/t3, v3/t1, v3/t3, t1/t4, t2/t4, t3/t4}
       \path[edge, black,  thick] (\source) --  (\dest);
       
       \end{scope}
       
\end{tikzpicture}
\end{center}
\end{minipage}\\
\\
$H_1$ & $H_2$ & $H_3$
\end{tabular}
\end{center}
\caption{The graphs used in Rule~\ref{r-8}.
A near-bipartite decomposition of each is indicated: the white vertices form an independent set and the black vertices induce a forest.}\label{fig:rule8}
\end{figure}
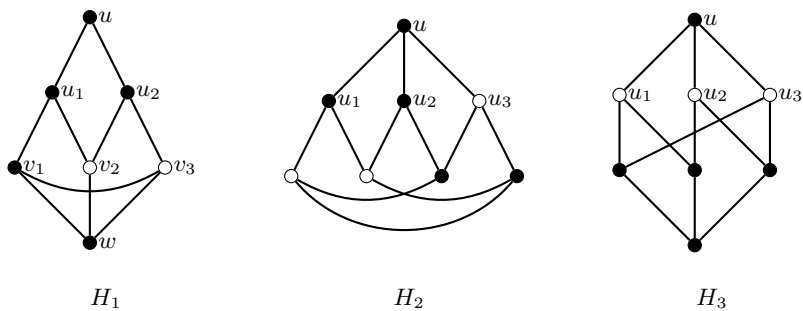

\noindent
Let us show that at least one of the rules is always applicable.
Suppose that, on the contrary, there is a vertex~$u$ of a subcubic graph for which no rule applies.
Then~$u$ and its neighbours each have degree~$3$ (Rules~\ref{r-1} and~\ref{r-2}) and so each either belongs to a triangle or is the centre of an induced claw.
By Rule~\ref{r-6}, $u$ must be the centre of an induced claw and therefore, by Rule~\ref{r-7}, the same is also true for each neighbour of~$u$.
This implies that Rule~\ref{r-9} applies, a contradiction.

Because~$G$ is subcubic, each of these rules takes constant time to verify and process.
In particular, in some rules we need to detect some induced subgraph of constant size that contains~$u$ or replace~$u$ by some other vertex~$v$.
In all such cases we need to explore a set of vertices of distance at most~$4$ from~$u$.
As~$G$ is subcubic, this set has size at most $1+3+3^2+3^3+3^4=121$, so we can indeed do this in constant time.

\medskip
\noindent
It is clear that, as claimed, the application of a rule reduces the number of vertices and that if we repeatedly choose an arbitrary vertex~$u$ and apply a rule, we eventually obtain the empty graph.
We now consider undoing the applied rules in reverse order to rebuild~$G$.
As this is done, we will irrevocably {\em colour} vertices with colour~$1$ or~$2$ in such a way that the vertices coloured~$1$ will form an independent set and the vertices coloured~$2$ will induce a forest.
Thus a rule is safe if this colouring can be extended whenever that rule is undone.
When we reach~$G$, the final colouring will correspond to the required near-bipartite decomposition.

We must prove each rule is safe.
At each step of reconstructing~$G$, we refer to the graph before a rule is undone as the \emph{prior graph} and to the graph after that rule is undone as the \emph{subsequent graph}.	
Note that the application of any of the Rules~\ref{r-1}--\ref{r-9} again yields a subcubic graph.
By the result of Yang and Yuan~\cite{YY06}, every connected subcubic graph is near-bipartite, apart from~$K_4$.
So we need to ensure that an application of a rule does not create a~$K_4$.
This cannot happen when we remove vertices, but we will need to consider it for Rules~\ref{r-6} and~\ref{r-9}.

\begin{figure}
\tikzstyle{vertex}=[circle,draw=black, fill=black, minimum size=5pt, inner sep=1pt]
\tikzstyle{edge} =[draw,-,black,>=triangle 90]
\begin{center}
\begin{tabular}{cc}
\begin{minipage}{0.45\textwidth}
\begin{center}
\begin{tikzpicture}
       
    \begin{scope}[yshift=-1cm]   
   \foreach \pos/\name / \label / \posn / \dist  in {
{(-1,0)/x/x/{above}/0},
{(0,1)/w/w/{above}/0},
{(1,0)/y/y/{above}/0},
{(0,-1)/v/v/{below}/0}}
      { \node[vertex] (\name) at \pos {};
       \node [\posn=\dist] at (\name) {$\label$};
       }
       
   \foreach \pos/\name / \label / \posn / \dist  in {
{(-2,0)/x'/x'/{left}/0},
{(2,0)/y'/y'/{right}/0}}
      { \node[vertex, fill=white] (\name) at \pos {};
       \node [\posn=\dist] at (\name) {$\label$};
       }

\foreach \source/ \dest  in {w/x, x/v, v/y, y/w, w/v}
       \path[edge, black,  thick] (\source) --  (\dest);     
       
\foreach \source/ \dest  in {x/x', y/y'}
       \path[edge, black!50!white,  thick, dashed] (\source) --  (\dest);

\end{scope}
\end{tikzpicture}
\end{center}
\end{minipage}
&
\begin{minipage}{0.45\textwidth}
\begin{center}
\begin{tikzpicture}
\begin{scope}[xshift=-2cm]
   \foreach \pos/\name / \label / \posn / \dist in {{(6.8,-1.5)/n2/y/{left}/0}, {(7.2,-0.8)/n3/x/{right}/0}, 
{(7.6,-1.5)/n5/u/{right}/0}}
       { \node[vertex] (\name) at \pos {};
       \node [\posn=\dist] at (\name) {$\label$};
       }
       
\foreach \source/ \dest  in {n2/n3, n2/n5, n3/n5}
       \path[edge, black,  thick] (\source) --  (\dest);

   \foreach \pos/\name / \label / \posn / \dist in {{(6.3,-2)/n1/y'/{left}/0},  
    {(7.2,0)/n4/x'/{right}/0},  {(8.1,-2)/n6/u'/{right}/0}}
       { \node[vertex, fill=white] (\name) at \pos {};
       \node [\posn=\dist] at (\name) {$\label$};
       }

\foreach \source/ \dest  in {n1/n2,  n3/n4, n5/n6}
       \path[edge, black!50!white,  thick, dashed] (\source) --  (\dest);

\end{scope}

\end{tikzpicture}
\end{center}
\end{minipage}\\
diamond & triangle
\end{tabular}
\end{center}
\caption{The diamond and triangle (solid edges and vertices) together with their neighbourhoods in a cubic graph.}
\label{fig:1}
\end{figure}
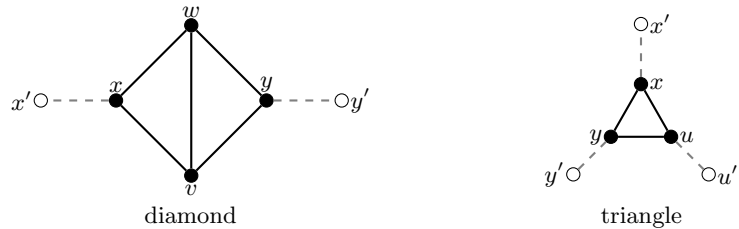

\clm{\label{clm:r1-r5-safe}Rules~\ref{r-1}--\ref{r-5} are safe.}
In Rules~\ref{r-1}--\ref{r-5} we only delete vertices.
Rule~\ref{r-1} is safe since if both neighbours of~$u$ are coloured~$2$, then~$u$ can be coloured~$1$; otherwise~$u$ can be coloured~$2$.
Similarly, we see that Rule~\ref{r-2} is safe.
To see that Rule~\ref{r-3} is safe, let~$D$ be the diamond with vertex labels as illustrated in \figurename~\ref{fig:1}, where~$u$ is one of $v,w,x,y$.
If~$x'$ and~$y'$ are coloured~$2$, we colour~$x$ and~$y$ with~$1$ and~$v$ and~$w$ with~$2$.
Otherwise we colour~$v$ with~$1$ and~$x$, $y$ and~$w$ with~$2$.
We now show that Rule~\ref{r-4} is safe.
Let~$u_1$ and~$u_2$ be false twins (at distance at most~$2$ from~$u$).
As~$G$ is subcubic, every vertex in~$N(u_1)$ with a neighbour in~$N(u_1)$ has no neighbours outside $N(u_1)\cup \{u_1,u_2\}$ and every vertex in~$N(u_1)$ with no neighbour in~$N(u_1)$ has at most one neighbour not equal to~$u_1$ or~$u_2$.
Moreover, as~$G$ is subcubic, $N(u_1)$ contains no cycle.
Hence we can always colour~$u_1$, $u_2$ with~$1$ and the vertices of~$N(u_1)$ with~$2$ regardless of the colours of vertices outside $N(u_1) \cup \{u_1, u_2\}$.
Indeed, every vertex of~$N(u_1)$ will have at most one neighbour that is not coloured~$1$, so cannot be in a cycle of vertices coloured~$2$ in the subsequent graph.
Rule~\ref{r-5} is also safe since~$P$ is $3$-regular and hence would be a component of our subsequent graph, so we can colour its vertices by assigning colour~$1$ to exactly one vertex from each of the two triangles and colour~$2$ to its other vertices (see \figurename~\ref{fig:prism-claw}).
This completes the proof of Claim~\ref{clm:r1-r5-safe}.

\clm{\label{clm:r6-r7-safe}Rules~\ref{r-6} and~\ref{r-7} are safe.}
First, let us demonstrate that Rule~\ref{r-6} is safe.
If~$x'$ and~$y'$ are contained in a~$K_4$ of the prior graph, then the subsequent graph contains a diamond whose vertices are at distance at most~$3$ from~$u$.
This contradicts Rule~\ref{r-3}.
Let~$T$ be the triangle with vertex labels as illustrated in \figurename~\ref{fig:1}.
Suppose~$x'$, $y'$ and~$u'$ are coloured~$2$.
Then we colour~$u$ with~$1$ and~$x$ and~$y$ with~$2$.
The vertices in the subsequent graph with colour~$2$ still induce a forest, as we have replaced an edge in the forest by a path on four vertices.
Suppose~$x'$ and~$y'$ are coloured~$2$ and~$u'$ is coloured~$1$.
Then we colour~$x$ with~$1$, and~$y$ and~$u$ with~$2$.
Otherwise, since~$x'$ and~$y'$ are joined by an edge in the prior graph, we may assume that~$x'$ has colour~$1$ and~$y'$ has colour~$2$.
In this case we can colour~$y$ with~$1$, and~$x$ and~$u$ with~$2$.
This completes the proof that Rule~\ref{r-6} is safe.
Since Rules~\ref{r-1}--\ref{r-6} are safe, it follows that Rule~\ref{r-7} is also safe.
This completes the proof of Claim~\ref{clm:r6-r7-safe}.

\clm{\label{clm:r8-safe}Rule~\ref{r-8} is safe.}
We now show that Rule~\ref{r-8} is safe.
Suppose~$u$ is contained in~$H_1$.
We use the vertex labels from \figurename~\ref{fig:rule8}.
As Rule~\ref{r-1} could not be applied, we find that~$u$ has a third neighbour~$u_3$ distinct from~$u_1$ and~$u_2$.
Regardless of whether~$u_3$ is coloured~$1$ or~$2$, we colour~$u$, $u_1$, $u_2$, $v_1$,~$w$ with~$2$ and~$v_2$, $v_3$ with~$1$ to obtain a near-bipartite decomposition of~$G$.
We can also readily colour the vertices of~$H_2$ or~$H_3$ should~$u$ be contained in one of them (note that since~$H_2$ and~$H_3$ are $3$-regular, these graphs can only appear as components in our subsequent graph).
This completes the proof of Claim~\ref{clm:r8-safe}.

\clm{\label{clm:r9-safe}Rule~\ref{r-9} is safe.}
Suppose that the prior graph contains fewer than three new edges.
Then we may assume without loss of generality that $e_1 = e_2$.
Then~$u_1$ and~$u_2$ are false twins at distance~$1$ from~$u$ and we can apply Rule~\ref{r-4}, a contradiction.
So we may assume that the prior graph contains exactly three new edges.

We claim that the application of Rule~\ref{r-9} does not yield a~$K_4$.
For contradiction, suppose it does.
Let~$K$ be the created~$K_4$.
Then at least one new edge is contained in~$K$.
If exactly one new edge~$e$ is contained in~$K$, then $K - e$ is a diamond in the subsequent graph.
Then we could have applied Rule~\ref{r-3}, a contradiction.
If all three new edges are in~$K$, then they must induce either a path on four vertices or a triangle in the subsequent graph.
In the first case the subsequent graph is~$H_2$ and in the second case the subsequent graph is~$H_3$.
In both cases we would have applied Rule~\ref{r-8}, a contradiction.
Finally, suppose that~$K$ contains exactly two new edges, say~$e_1$ and~$e_2$.
If~$e_1$ and~$e_2$ do not share a vertex, then they cover the vertices of~$K$.
Hence the end-vertices of~$e_1$ are false twins (at distance~$2$ from~$u$) in the subsequent graph, since they are both adjacent to~$u_1$ and to each end-vertex of~$e_2$.
Then we could have applied Rule~\ref{r-4}, a contradiction.
If $e_1 =v_1v_2$ and $e_2 = v_3v_4$ share a vertex, say $v_2 = v_4$, then~$v_1$ and~$v_3$ are adjacent in the subsequent graph and the vertex $w \in K \setminus \{v_1,v_2,v_3\}$ is adjacent only to~$v_1$, $v_2$ and~$v_3$.
Therefore Rule~\ref{r-8} could have been applied, a contradiction.

Thus an application of Rule~\ref{r-9} does not yield a~$K_4$, and we may colour~$u_1$, $u_2$,~$u_3$ with~$2$ and~$u$ with~$1$.
Indeed note that since if two end-vertices of a new edge are coloured~$2$, then in the subsequent graph the vertices coloured~$2$ will still induce a forest, in which such a new edge is replaced by a path of length~$2$.
This completes the proof of Claim~\ref{clm:r9-safe} and therefore completes the proof of Theorem~\ref{thm:cubic}.\qedllncs
\end{proof}

\section{Quadratic Time for Graphs of Maximum Degree at Least~$3$}\label{s-maxdegree}
Let $k\geq 3$ be an integer.
Recall that a graph~$G$ has a $k$-degenerate decomposition if its vertex set can be decomposed into sets~$A$ and~$B$ where~$A$ is an independent set and~$B$ induces a $(k-2)$-degenerate graph.
Note that $3$-degenerate decompositions are near-bipartite decompositions.
We give an~$O(n^2)$ algorithm for finding a $k$-degenerate decomposition of a graph on~$n$ vertices of maximum degree at most~$k$ for every $k\geq 3$ (note that for $k=3$ we can also use Theorem~\ref{thm:cubic}).

For $k \geq 1$, we say that an order $v_1,v_2,\ldots,v_n$ of the vertices of a graph~$G$ is \emph{$k$-degenerate} if for all $i \geq 2$, the vertex~$v_i$ has at most~$k$ neighbours in $\{v_1,\ldots,v_{i-1}\}$.
It is clear that a graph is $k$-degenerate if and only if it has a $k$-degenerate order.
If~${\mathcal O}$ is a $k$-degenerate order for~$G$ and~$W$ is a subset of the vertex set of~$G$, then we let~$\mathcal{O}|_W$ be the restriction of~${\mathcal O}$ to~$W$, and let~$G[W]$ be the subgraph of~$G$ induced by~$W$.
For a set of vertices~$C$, we denote the {\em neighbourhood} of~$C$ by $N(C)=\bigcup\{N(u) \; | \; u \in C\} \setminus C$.

We need the following lemma, which is a refinement of Lemma~8 in~\cite{FJP16} with the same proof.
We include a sketch of the argument for completeness.
\begin{lemma}\label{lem:degenerate}
Let $k \geq 2$.
Let~$G$ be a $(k-1)$-degenerate graph on~$n$ vertices.
If a $(k-1)$-degenerate order~$\mathcal{O}$ of~$G$ is given as input, a $k$-degenerate decomposition $(A,B)$ of~$G$ can be found in~$O(kn)$ time.
In addition, we can ensure that~$\mathcal{O}|_B$ is a $(k-2)$-degenerate order of~$G[B]$, the set~$A$ is a maximal independent set and the first vertex in~$\mathcal{O}$ belongs to~$A$.
\end{lemma}

\begin{proof}
Let~$\mathcal{O}$ be $v_1,v_2,\ldots,v_n$.
Consider the greedy algorithm that starts with two empty sets~$A$ and~$B$, and, at step~$i$, assigns~$v_i$ to~$A$ unless~$v_i$ has a neighbour of smaller index already in~$A$, in which case it assigns~$v_i$ to~$B$.
Clearly the set~$A$ is an independent set and every vertex of~$B$ has at least one neighbour in~$A$, so~$A$ is a maximal independent set.
At any step~$i$, if the vertex~$v_i$ is assigned to~$B$, then it has a neighbour of smaller index that belongs to~$A$.
This implies that~$v_i$ has at most~$k-2$ neighbours of smaller index that belong to~$B$.\qedllncs
\end{proof}

A pair of non-adjacent vertices $\{u,v\}$ in a graph~$G$ is \emph{strong} if~$u$ and~$v$ have a common neighbour in each component of the graph $G \setminus \{u, v\}$.
In particular, note that if~$u$ and~$v$ have a common neighbour and $G \setminus \{u, v\}$ is connected, then $\{u,v\}$ is a strong pair.

\begin{lemma}\label{lem:ifstrongpair}
Let $k \geq 3$.
Let~$G$ be a connected $k$-regular graph on~$n$ vertices that contains a strong pair~$\{u,v\}$.
If~$\{u,v\}$ is given as input, a $k$-degenerate decomposition of~$G$ can be found in~$O(kn)$ time.
\end{lemma}

\begin{proof}
Let~$G'$ be the graph obtained by identifying~$u$ and~$v$ into a new vertex~$z$ with $N(z)=N(u) \cup N(v)$.
For each component~$C$ of $G' \setminus \{z\}$, let~$z_C$ be a vertex of~$C$ that is adjacent to both~$u$ and~$v$ in~$G$.
We find a $(k-1)$-degenerate order~${\mathcal O}$ of the vertices of~$G'$ as follows.
Let~$z$ be the first vertex in~${\mathcal O}$.
Consider each component~$C$ of $G' \setminus \{z\}$ in turn, and append to~${\mathcal O}$ the vertices of~$C$ in the reverse of the order they are found in a breadth-first search from~$z_C$ (note that this can be done in $O(kn)$ time).
Then~${\mathcal O}$ is a $(k-1)$-degenerate order as every vertex has a neighbour later in the order except for each~$z_C$ vertex, which has degree $k-1$.
It follows from Lemma~\ref{lem:degenerate} that we can find a $k$-degenerate decomposition $(A, B)$ of~$G'$ with $z \in A$.
Then $(A \setminus \{z\} \cup \{u,v\}, B)$ is a $k$-degenerate decomposition of~$G$ as $G[B]=G'[B]$, and~$u$ and~$v$ are non-adjacent so $A \setminus \{z\} \cup \{u,v\}$ is an independent set.\qedllncs
\end{proof}

\begin{lemma}\label{lem:ifquasiclique}
Let $k \geq 3$.
Let~$G$ be a $k$-regular connected graph on~$n$ vertices, which contains a set~$C$ of $k+1$ vertices that induce a clique minus an edge~$uv$.
If $C$, $u$ and~$v$ are given as input, then a $k$-degenerate decomposition of~$G$ can be found in~$O(kn)$ time.
\end{lemma}
\begin{proof}
Let~$x$ be a vertex in~$C$ distinct from~$u$ and~$v$.
Let~$G'$ be the graph obtained from~$G$ by deleting~$C$.
Each of~$u$ and~$v$ has exactly one neighbour that does not belong to~$C$ and all other vertices of~$C$ have no neighbours outside~$C$.
Let~$t$ be the neighbour of~$u$ not in~$C$, and let~$w$ be the neighbour of~$v$ not in~$C$.
We may assume that~$t$ is distinct from~$w$, otherwise we are done by Lemma~\ref{lem:ifstrongpair}.
We can find a $(k-1)$-degenerate order~${\mathcal O}$ of~$G'$ in $O(kn)$ time by taking the vertices in the reverse of the order they are found in a breadth-first search from~$t$, and then, if~$t$ and~$w$ do not belong to the same component of~$G'$, appending the vertices in the reverse of the order they are found in a breadth-first search from~$w$.
By Lemma~\ref{lem:degenerate}, we can compute a $k$-degenerate decomposition~$(A,B)$ of~$G'$ in~$O(kn)$ time such that~${\mathcal O}|_B$ is a $(k-2)$-degenerate order of~$B$.
If both~$t$ and~$w$ belong to~$B$, let $A'=A \cup \{u,v\}$ and $B'= B \cup (C \setminus \{u,v\})$ and, since $(C \setminus \{u,v\})$ is a clique on $k-1$ vertices with no edge joining it to~$B$, if follows that $(A',B')$ is a $k$-degenerate decomposition of~$G$.
Assume now without loss of generality that $t \in A$ (we make no assumption about whether~$w$ is also in~$A$).
Then let $A'=A \cup \{x\}$ and $B'=B \cup (C \setminus \{x\})$.
Then~$A'$ is an independent set.
Recall that~${\mathcal O}|_B$ is a $(k-2)$-degenerate order of~$B$.
We show that we can append the vertices of $C \setminus \{x\}$ to obtain a $(k-2)$-degenerate order of~$B'$.
First add~$v$, then the vertices of $C \setminus \{u,v,x\}$ and finally~$u$.
It is clear that no vertex has more than $k-2$ neighbours earlier in the order.\qedllncs
\end{proof}

\begin{lemma}\label{lem:ifclique}
Let $k\geq 3$.
Let~$G$ be a $k$-regular connected graph on~$n$ vertices containing a clique~$C$ on~$k$ vertices whose neighbourhood is of size~$2$.
If~$C$ is given as input, a $k$-degenerate decomposition of~$G$ can be found in~$O(kn)$ time.
\end{lemma}
\begin{proof}
Let~$u$ and~$v$ be vertices not in~$C$ such that for each vertex in~$C$, its unique neighbour not in~$C$ is either~$u$ or~$v$.
Neither~$u$ nor~$v$ can be adjacent to every vertex in~$C$ (as then the other would be adjacent to none, contradicting the premise that the neighbourhood has size~$2$).
Since $k \geq 3$, one of~$u$ and~$v$ has at least two neighbours in~$C$.
Consider the graph~$G'$ obtained from~$G$ by removing~$C$ and adding the edge~$uv$ (if it does not already exist).
Note that~$u$ and~$v$ each have degree at most~$k$ in~$G'$ and at least one of them, say~$u$, has degree less than~$k$.
Therefore, we can find a $(k-1)$-degenerate order~${\mathcal O}$ of~$G'$ in $O(kn)$ time by taking the vertices in the reverse of the order they are found in a breadth-first search from~$u$.
Thus we can obtain a $k$-degenerate decomposition $(A,B)$ of~$G'$ by Lemma~\ref{lem:degenerate}, such that~$\mathcal{O}|_B$ is a $(k-2)$-degenerate order of~$G[B]$ and~$A$ is a maximal independent set.

To prove the lemma, we must add each vertex of $C$ to either $A$ or $B$ to obtain sets~$A'$ amd $B'$ such that $(A',B')$ is a $k$-degenerate decomposition  of~$G$.  We note that at most one of $u$ and $v$ is in $A$.  If $u \in A$, then let $t$ be a vertex of $C$ not adjacent to $u$ (and so $t$ is adjacent to $v$).  If $u \not \in A$, then let $t$ be a vertex of $C$ adjacent to $u$ (and not~$v$).   Let  $A'=A \cup \{t\}$ and $B'=B \cup (C \setminus \{t\})$,  We claim that $(A',B')$ is a $k$-degenerate decomposition of~$G$.  It is clear that~$A'$ is an independent set.
Recall that~${\mathcal O}|_B$ is a $(k-2)$-degenerate order for~$G[B]$.
We must amend it to find a $(k-2)$-degenerate order for~$G[B']$ that also includes the vertices of $C \setminus \{t\}$.
We consider three cases.

First suppose that neither $u$ nor $v$ is in $A$.
Then~$u$ has a neighbour in~$G'$ that belongs to~$A$ (as~$A$ is a maximal independent set) and is also a neighbour of $t \in A'$.
Hence~$u$ has at most $k-2$ neighbours in~$B'$.
Append to~${\mathcal O}|_B$ the vertices of $C \setminus \{t\}$, ending with a neighbour of~$u$ (we know there is at least one), then move~$u$ to be the final vertex in the order.
The only vertex of $C \setminus \{t\}$ that could have more than $k-2$ neighbours before it in the order is the one that appears last (since they are all adjacent to $t \not \in B'$), and by choosing it to be a neighbour of~$u$ and putting~$u$ later in the order we ensure that a $(k-2)$-degenerate order is obtained.

Now suppose $v \in A$.
Then append to~${\mathcal O}|_B$ first the neighbours of~$u$ in $C \setminus \{t\}$ and then the neighbours of~$v$ (and we know there is at least one such neighbour in $C$ and such a vertex was not chosen to be $t$).
Again, as the vertices of $C \setminus \{t\}$ are adjacent to~$t$ the only one that could have more than $k-2$ vertices before it in the order is the one that appears last, but this is also adjacent to~$v$ (which is not in $B$) so we do indeed have a $(k-2)$-degenerate order.

Finally suppose $u \in A$.
Then append to~${\mathcal O}|_B$ first the neighbours of~$v$ in $C \setminus \{t\}$ and then the neighbours of~$u$ (and we know there are such neighbours in $C$ and none were chosen to be $t$).
Again the only one of these vertices that could have more than $k-2$ vertices before it in the order is the last one, but this is also adjacent to~$u$ (which is not in $B$) so we do indeed have a $(k-2)$-degenerate order.
\qedllncs
\end{proof}

Given a graph~$G$, five of its vertices $t,u,v,w,x$ and a set of vertices~$C$, we say that~$C$ induces a \emph{$(u,v)$-lock with special vertices $(t,\{w,x\})$} if $t,w,x \in C$ and $N(C)=\{u,v\}$, and both~$u$ and~$v$ are adjacent to~$t$, each vertex in $\{w,x\}$ is adjacent to precisely one vertex in $\{u,v\}$, and $G[C]$ contains all possible edges except for~$wt$ and~$xt$.
We say that~$C$ is a \emph{lock} if it is a \emph{$(u,v)$-lock with special vertices $(t,\{w,x\})$} for some choice of $t,u,v,w,x$.

\begin{lemma}\label{lem:iflock}
Let $k\geq 3$.
Let~$G$ be a $k$-regular connected graph on~$n$ vertices containing a $(u,v)$-lock~$C$ with special vertices $(t,\{w,x\})$.
If~$C$ and $u,v,t,w,x$ are given as input, then a $k$-degenerate decomposition of~$G$ can be found in~$O(kn)$ time.
\end{lemma}

\begin{proof}
Since~$t$ has two neighbours outside~$C$, it follows that~$C\setminus\{t\}$ is a clique on~$k$ vertices.
If~$w$ and~$x$ have the same neighbour in~$\{u,v\}$, say~$u$, then $N(C \setminus \{t\})=\{t,u\}$ and so we are done by Lemma~\ref{lem:ifclique}.
We may therefore assume that~$w$ and~$x$ have distinct neighbours in~$\{u,v\}$.
Let~$G'$ be the graph obtained from~$G$ by deleting~$C\setminus \{t\}$ and note that~$G'$ is connected since~$t$ is adjacent to both~$u$ and~$v$.
Note that both~$u$ and~$v$ have degree $k-1$ in~$G'$.
We can therefore find a $(k-1)$-degenerate order~${\mathcal O}$ of~$G'$ in $O(kn)$ time by taking the vertices in the reverse of the order they are found in a breadth-first search from~$u$.
Furthermore, since the only neighbours of~$t$ in~$G'$ are~$u$ and~$v$, both of which have degree $k-1$, by moving~$t$ to the start of the order~${\mathcal O}$, we obtain a another $(k-1)$-degenerate order~${\mathcal O'}$.
By Lemma~\ref{lem:degenerate}, we can therefore find a $k$-degenerate decomposition $(A,B)$ of~$G'$ such that~${\mathcal O'}|_B$ is a $(k-2)$-degenerate order on~$B$ and $t \in A$.
Thus both~$u$ and~$v$ belong to~$B$.
We let $A'=A\cup\{w\}$ and $B'=B \cup (C \setminus \{w,t\})$, and claim that $(A',B')$ is a $k$-degenerate decomposition of~$G$.
It is clear that~$A'$ is an independent set.
We have the $(k-2)$-degenerate order~${\mathcal O'}|_B$ on~$B$.
We obtain a $(k-2)$-degenerate order on~$B'$ by appending to~${\mathcal O'}|_B$ the vertices of $C \setminus \{w,t\}$ beginning with~$x$.
Indeed, the only neighbour of~$x$ that is earlier in the order is its single neighbour in $\{u,v\}$ (and note that $1 \leq k-2$ since $k \geq 3$).
Furthermore, since $w,t \in A'$, every vertex in $C \setminus \{w,t,x\}$ has only $k-2$ neighbours in~$B'$.\qedllncs
\end{proof}

A pair of non-adjacent vertices~$u$, $v$ in a graph is a \emph{good} pair if~$u$ and~$v$ have a common neighbour.
Note that if a good pair~$u$, $v$ is not strong, then $G \setminus \{u,v\}$ must be disconnected.
We are now ready to state and prove the following result.

\begin{theorem}\label{thm:k2-deg}Let $k\geq 3$ and~$G$ be a graph on~$n$ vertices with maximum degree at most~$k$.
If no component of~$G$ is isomorphic to~$K_{k+1}$, then a $k$-degenerate decomposition of~$G$ can be found in~$O(kn^2)$ time.
\end{theorem}

\begin{proof}
\hyphenation{component-wise}
We may assume that~$G$ is connected, otherwise it can be considered componentwise.
If~$G$ is not $k$-regular, then it has a vertex~$u$ of degree at most $k-1$, so we can find a $(k-1)$-degenerate order~${\mathcal O}$ of~$G$ in $O(kn)$ time by taking the vertices in the reverse of the order they are found in a breadth-first search from~$u$.
In this case, we are done by Lemma~\ref{lem:degenerate}.
For $k$-regular graphs, we use the procedure shown in Algorithm~\ref{algo:1}.
We note that if~$G$ is biconnected then, by~\cite[Lemma~3]{BW14}, there is always a good pair~$u$, $v$ such that $G \setminus \{u,v\}$ is connected, but this does not aid us in finding an algorithm for general graphs.

\begin{algorithm}
\LinesNumbered
\SetKw{KwGoTo}{go to}
\SetKwInOut{Input}{Input}
\SetKwInOut{Output}{Output}
\Input{A connected~$k$-regular graph~$G$}
\Output{A $k$-degenerate decomposition of~$G$}
\BlankLine
\emph{\nllabel{line:1}find a good pair~$u$, $v$ and let~$C$ be a component of $G \setminus \{u,v\}$}\;
\lIf{\nllabel{line:2}$u$, $v$ is a strong pair}{apply Lemma~\ref{lem:ifstrongpair}}
\lElseIf{\nllabel{line:4}the union of~$C$ and one or both of~$u$ and~$v$ is a clique on $k+1$ vertices minus an edge}{apply Lemma~\ref{lem:ifquasiclique}}
\lElseIf{\nllabel{line:5}$C$ is a clique on~$k$ vertices whose neighbourhood is $\{u,v\}$}{apply Lemma~\ref{lem:ifclique}}
\lElseIf{\nllabel{line:6}$C$ is a $(u,v)$-lock}{apply Lemma~\ref{lem:iflock}}
\Else{\emph{\nllabel{line:8}find a good pair $u', v' \in C$ such that either $C'=G \setminus \{u',v'\}$ is connected or $G \setminus \{u',v'\}$ has a component~$C'$ that is strictly contained in~$C$}\;
\nllabel{line:12}set $u \leftarrow u'$, $v \leftarrow v'$, $C \leftarrow C'$\;
\KwGoTo{\nllabel{line:13}Line~\ref{line:2}}
}
\caption{\label{algo:1}Finding a $k$-degenerate decomposition for connected $k$-regular graphs.}
\end{algorithm}

Let us make a few comments on this procedure.
As~$G$ is regular, connected and not complete, we can initially choose any vertex as~$u$ and find another vertex~$v$ to form a good pair in $O(k^2)=O(k n)$ time.
If we perform a breadth-first search (which takes $O(n+m)=O(kn)$ time) from a neighbour of~$u$ that retreats from~$u$ or~$v$ whenever they are encountered, we discover a component of $G \setminus \{u, v\}$.
If the component contains a common neighbour of~$u$ and~$v$ but is not equal to $G \setminus \{u,v\}$, we repeat starting from a neighbour of~$u$ or~$v$ that was not discovered.
Thus we discover in~$O(kn)$ time that either~$u$, $v$ is a strong pair (if we find all components of $G \setminus \{u,v\}$ and they each contain a common neighbour of~$u$ and~$v$), or that it is not.
We set~$C$ to be one of the components of $G \setminus \{u, v\}$ arbitrarily.
By Lemma~\ref{lem:ifstrongpair}, we therefore conclude that Lines~\ref{line:1} and~\ref{line:2} take $O(kn)$ time.
It is easy to check in $O(kn)$ time whether we apply Lemmas~\ref{lem:ifquasiclique}--\ref{lem:iflock} on Lines~\ref{line:4}--\ref{line:6} and applying these lemmas takes $O(kn)$ in each case.
Now suppose that we do not apply any of these lemmas, in which case we reach Line~\ref{line:8}.
We will show that we can find~$u'$, $v'$ and, if necessary,~$C'$ in~$O(kn)$ time.
If we find~$u'$,~$v'$ such that $C'=G \setminus \{u',v'\}$ is connected, then $u',v'$ is a strong pair, so after executing Line~\ref{line:13}, the algorithm will stop on Line~\ref{line:2}.
In all other cases~$C'$ will be strictly smaller than~$C$.
This means that we apply Line~\ref{line:13} at most~$O(n)$ times, implying that we execute Lines~\ref{line:2}--\ref{line:13} at most~$O(n)$ times.
This will give an overall running time of $O(kn^2)$.
It remains to show that if execution reaches Line~\ref{line:8} then we can find the required good pair $u',v'$ and the component~$C'$ in~$O(kn)$ time.

Let us first show that~$C$ contains good pairs --- that is, that it is not a clique.
If~$C$ is a clique, then it contains either $k-1$ or~$k$ vertices (as each vertex has degree~$k$ in~$G$ and the only other possible neighbours of vertices in~$C$ are~$u$ and~$v$).
If~$C$ has $k-1$ vertices, then each vertex of~$C$ must be adjacent to both~$u$ and~$v$ and we would have applied Lemma~\ref{lem:ifquasiclique} on Line~\ref{line:4}, a contradiction.
If~$C$ has~$k$ vertices, then each vertex of~$C$ is adjacent to exactly one of~$u$ and~$v$ (and neither~$u$ nor~$v$ can be adjacent to every vertex in~$C$, as this would form a~$K_{k+1}$, contradicting the fact that~$G$ is connected), in which case we would have applied Lemma~\ref{lem:ifclique} on Line~\ref{line:5}, a contradiction.
Therefore we may assume that~$C$ is not a clique.

We need describe how to choose a good pair~$u'$, $v'$ in~$C$.
If we can show that~$u$ and~$v$ are in the same component of $G \setminus \{u',v'\}$ (which must necessarily contain all of $G \setminus C$), then we are done as either $G\setminus \{u',v'\}$ is connected or there is another component~$C'$ of $G \setminus \{u',v'\}$ which must be contained in~$C$ (and note that in this case~$C'$ can be found in $O(kn)$ time using breadth-first search).

If~$u$ and~$v$ have a common neighbour outside~$C$ or at least three common neighbours in~$C$, then any good pair in~$C$ can be chosen as~$u'$, $v'$ (as~$u$ and~$v$ will then be in the same component of $G \setminus \{u',v'\}$).
If~$u$ and~$v$ have exactly two common neighbours~$t_1$, $t_2$ that both belong to~$C$, then any good pair other than~$t_1$, $t_2$ can be chosen as~$u'$, $v'$.
If~$t_1$, $t_2$ is the only good pair in~$C$ (so all other vertices in~$C$ are adjacent), then~$C$ is a clique minus an edge and must contain~$k$ vertices ($t_1$,~$t_2$ and the~$k-2$ neighbours of~$t_1$ that are not in $\{u,v\}$).
Considering degree, any vertex in~$C$ other than~$t_1$ or~$t_2$ must be adjacent to exactly one of~$u$ and~$v$.
If every vertex in $C \setminus \{t_1, t_2\}$ is adjacent to, say~$u$, then $C \cup \{u\}$ is a clique on $k+1$ vertices minus an edge and we would have applied Lemma~\ref{lem:ifquasiclique} on Line~\ref{line:4}, a contradiction.
We may therefore assume that at least one vertex in $C \setminus \{t_1,t_2\}$ is adjacent to~$u$ and at least one is adjacent to~$v$, so there is a path from~$u$ to~$v$ avoiding~$t_1$ and~$t_2$, and $G \setminus \{t_1, t_2\}$ is connected, so we are done.

Finally, suppose that~$u$ and~$v$ have exactly one common neighbour~$t$ that belongs to~$C$.
Then any good pair not including~$t$ can be chosen as~$u'$, $v'$, as then~$u$ and~$v$ will be in the same component of $G \setminus \{u',v'\}$.
Suppose, for contradiction, that no such pair exists.
Then $C \setminus \{t\}$ is a clique.
The vertex~$t$ has $k-2$ neighbours in $C\setminus \{t\}$.
Since $k\geq 3$, let~$z$ be one of those neighbours.
Since~$t$ is the only common neighbour of~$u$ and~$v$, we have that~$z$ can only be adjacent to at most one of~$u$ and~$v$.
Therefore, $t$ has a neighbour non-adjacent to~$z$, so~$z$ must have a neighbour non-adjacent to~$t$, which we denote~$w$.
As~$w$ is also adjacent to at most one of~$u$ and~$v$, it also has a neighbour~$x$ that is a non-neighbour of~$t$ (and cannot be~$t$ itself).
So $C \setminus \{t\}$ contains at least~$k$ vertices: the $k-2$ neighbours of~$t$ plus~$w$ and~$x$.
As $C \setminus \{t\}$ induces a clique, it must have exactly~$k$ vertices, since~$G$ cannot contain a~$K_{k+1}$.
Thus the set~$C$ forms a lock, and so we would have applied Lemma~\ref{lem:iflock} on Line~\ref{line:6}.
This contradiction completes the proof.\qedllncs
\end{proof}

Theorems~\ref{thm:cubic} and~\ref{thm:k2-deg} provide an algorithmic version of Theorem~\ref{t-special}.
Moreover, Theorems~\ref{thm:cubic} and~\ref{thm:k2-deg} concern decompositions $(A,B)$ of the vertex set of a graph where~$A$ is independent and~$B$ induces a $(k-2)$-degenerate graph.
As~$B$ therefore cannot be a clique on~$k$ vertices, both theorems also provide an algorithmic version of Theorem~\ref{t-ca79}.

\section{Hardness for Graphs of Large Maximum Degree}\label{s-hard}

Recall that Yang and Yuan~\cite{YY06} showed that the problem of deciding whether a graph of maximum degree~$4$ is near-bipartite, or equivalently, has a $3$-degenerate decomposition is \NP-complete.
We show how their proof can be adapted to prove that for every $k\geq 3$ the problem of deciding whether a graph of maximum degree~$2k-2$ has a $k$-degenerate decomposition is \NP-complete.

We start by explaining the proof of Yang and Yuan.
They use a reduction from {\sc $1$-in-$3$-Sat} with positive literals only.
That is, we are given a set of variables $X=\{x_1,\ldots,x_n\}$ and a set of clauses ${\cal C}=\{C_1,\ldots,C_m\}$ such that each clause~$C_i$ is of the form $(x_a \vee x_b \vee x_c)$ for some $a,b,c$ with $1\leq a<b<c\leq n$, and the question is whether there exists a truth assignment that makes exactly one of the three variables in each clause true (and the other two false).

Given an instance of {\sc $1$-in-$3$-Sat}, Yang and Yuan construct a graph~$G$ as follows.
First, as a building block they introduce a graph~$H(a,b,c,d,e)$ that consists of five vertices $a,b,c,d,e$ with edges $ab$, $ac$, $bc$, $bd$, $be$, $cd$, $ce$ (see also \figurename~\ref{fig:H0}).
They define the graph~$H^k$ recursively a follows:
\begin{itemize}
\item $H^0=H(a_1^0,b_1^0,c_1^0,d_1^0,e_1^0)$.
\item For $i \geq 1$, construct~$H^{i-1}$ and~$2^i$ copies of~$H$, that is, for $1 \leq j \leq 2^i$, construct $H(a_j^i,b_j^i,c_j^i,d_j^i,e_j^i)$.
Next, for $1 \leq j \leq 2^{i-1}$ identify the vertices~$d_j^{i-1}$ and~$a_{2j-1}^i$ and identify the vertices $e_j^{i-1}$ and~$a_{2j}^i$.
\end{itemize}
See \figurename~\ref{fig:H2} for the graph~$H^2$.
Observe that the graph~$H^i$ contains~$2^i$ pairwise vertex-disjoint copies of~$H$ that do not have their~$d$ and~$e$ vertices identified with~$a$ vertices.

\begin{figure}
\tikzstyle{vertex}=[circle,draw=black, fill=black, minimum size=5pt, inner sep=1pt]
\tikzstyle{edge} =[draw,-,black,>=triangle 90]
\begin{center}
\begin{tikzpicture}
\foreach \pos/\name / \label / \posn / \dist in {
{(0,0)/a/a/{above}/0},
{(-1,-1)/b/b/{left}/0},
{(1,-1)/c/c/{right}/0},
{(-1,-2)/d/d/{left}/0},
{(1,-2)/e/e/{right}/0}}
{\node[vertex] (\name) at \pos {};
\node [\posn=\dist] at (\name) {$\label$};
}

\foreach \source/ \dest in {a/b, b/c, c/a, b/d, b/e, c/d, c/e}
\path[edge, black,  thick] (\source) --  (\dest);
\end{tikzpicture}
\end{center}
\caption{\label{fig:H0}The graph~$H(a,b,c,d,e)$.}
\end{figure}
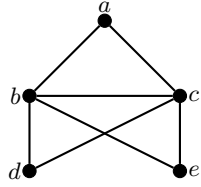

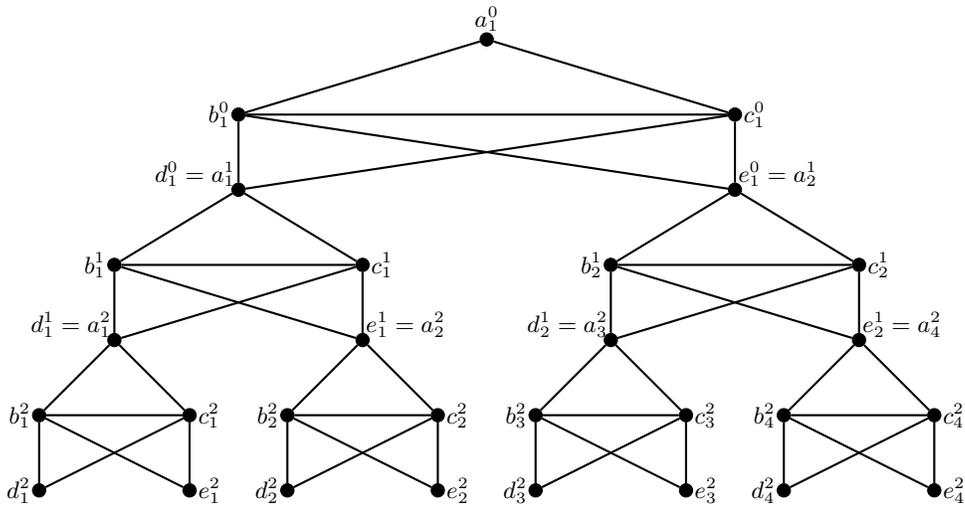
\begin{figure}
\tikzstyle{vertex}=[circle,draw=black, fill=black, minimum size=5pt, inner sep=1pt]
\tikzstyle{edge} =[draw,-,black,>=triangle 90]
\begin{center}
\begin{tikzpicture}
\foreach \pos/\name / \label / \posn / \dist in {
{(0,0)/a/a_1^0/{above}/0},
{(-3.3,-1)/b/b_1^0/{left}/0},
{(3.3,-1)/c/c_1^0/{right}/0},
{(-3.3,-2)/d/d_1^0=a_1^1/{above left}/-0.1},
{(3.3,-2)/e/e_1^0=a_2^1/{above right}/-0.1}}
{\node[vertex] (\name) at \pos {};
\node [\posn=\dist] at (\name) {$\label$};
}
\foreach \source/ \dest in {a/b, b/c, c/a, b/d, b/e, c/d, c/e}
\path[edge, black,  thick] (\source) --  (\dest);

\foreach \num/\shift in {1/d}{
\begin{scope}[shift={(\shift)}]
\foreach \pos/\name / \label / \posn / \dist in {
{(-1.65,-1)/b1\num/b^1_\num/{left}/0},
{(1.65,-1)/c1\num/c^1_\num/{right}/0},
{(-1.65,-2)/d1\num/d^1_\num=a_1^2/{above left}/-0.1},
{(1.65,-2)/e1\num/e^1_\num=a_2^2/{above right}/-0.1}}
{\node[vertex] (\name) at \pos {};
\node [\posn=\dist] at (\name) {$\label$};
}

\foreach \source/ \dest in {\shift/b1\num , \shift/c1\num, b1\num/c1\num, b1\num/d1\num, b1\num/e1\num, c1\num/d1\num, c1\num/e1\num}
\path[edge, black,  thick] (\source) --  (\dest);

\end{scope}
}

\foreach \num/\shift in {2/e}{
\begin{scope}[shift={(\shift)}]
\foreach \pos/\name / \label / \posn / \dist in {
{(-1.65,-1)/b1\num/b^1_\num/{left}/0},
{(1.65,-1)/c1\num/c^1_\num/{right}/0},
{(-1.65,-2)/d1\num/d^1_\num=a_3^2/{above left}/-0.1},
{(1.65,-2)/e1\num/e^1_\num=a_4^2/{above right}/-0.1}}
{\node[vertex] (\name) at \pos {};
\node [\posn=\dist] at (\name) {$\label$};
}

\foreach \source/ \dest in {\shift/b1\num , \shift/c1\num, b1\num/c1\num, b1\num/d1\num, b1\num/e1\num, c1\num/d1\num, c1\num/e1\num}
\path[edge, black,  thick] (\source) --  (\dest);

\end{scope}
}

\foreach \num/\shift in {1/d11,2/e11,3/d12,4/e12}{
\begin{scope}[shift={(\shift)}]
\foreach \pos/\name / \label / \posn / \dist in {
{(-1,-1)/b2\num/b^2_\num/{left}/0},
{(1,-1)/c2\num/c^2_\num/{right}/0},
{(-1,-2)/d2\num/d^2_\num/{left}/0},
{(1,-2)/e2\num/e^2_\num/{right}/0}}
{\node[vertex] (\name) at \pos {};
\node [\posn=\dist] at (\name) {$\label$};
}

\foreach \source/ \dest in {\shift/b2\num , \shift/c2\num, b2\num/c2\num, b2\num/d2\num, b2\num/e2\num, c2\num/d2\num, c2\num/e2\num}
\path[edge, black,  thick] (\source) --  (\dest);

\end{scope}
}

\end{tikzpicture}
\end{center}
\caption{\label{fig:H2}The graph~$H^2$.}
\end{figure}

Yang and Yuan prove the following observations for $i\geq 1$:

\begin{enumerate}[(i)]
\renewcommand{\theenumi}{(\roman{enumi})}
\renewcommand{\labelenumi}{(\roman{enumi})}
\item\label{prop:i}the graph~$H^i$ is near-bipartite;
\item\label{prop:ii}in any near-bipartite decomposition of~$H^i$ either all vertices~$d^i_j$ and~$e^i_j$ ($j \in \{1,\ldots,2^i\}$) belong to~$A$ or they all belong to~$B$;
\item\label{prop:iii}in any near-bipartite decomposition of~$H^i$ at least one of the vertices $b^i_j$, $c^i_j$ ($j \in \{1,\ldots,2^i\}$) belongs to $B$.
\end{enumerate}

To finish off the construction of the graph~$G$, Yang and Yuan, take~$n$ copies $H^q[x_1],\ldots,H^q[x_n]$ of the graph~$H^q$, where~$q$ is chosen such that $2^{q-1}<m\leq 2^q$, so that each copy of~$H^q[x_i]$ corresponds to the variable~$x_i$.
To distinguish variables in different copies of~$H^q$, we will refer, for example, to the copy of~$e^i_j$ in~$H^q[x_h]$ as~$e^i_j[x_h]$.
For each clause $C_j=(x_a \vee x_b \vee x_c)$ they introduce a {\em clause-triangle}, that is, a triangle with vertices $x_a^j$, $x_b^j$, $x_c^j$.
If variable~$x_h$ occurs in a clause~$C_j$, then they add an edge between~$x_h^j$ and~$d^q_j[x_h]$ and an edge between~$x_h^j$ and~$e^q_j[x_h]$.
Note that $G$ has maximum degree~$4$.

Next, Yang and Yang prove the following three observations in the case where~$G$ is near-bipartite:

\begin{enumerate}[(i)]
\renewcommand{\theenumi}{(\roman{enumi})}
\renewcommand{\labelenumi}{(\roman{enumi})}
\setcounter{enumi}{3}
\item\label{prop:iv}exactly one of the vertices of each clause-triangle belongs to~$A$;
\item\label{prop:v}if a vertex~$x_h^j$ belongs to~$A$, then every other~$x_h^{j^*}$ also belongs to~$A$; and
\item\label{prop:vi}if a vertex~$x_h^j$ belongs to~$B$, then every other~$x_h^{j^*}$ also belongs to~$B$.
\end{enumerate}

\noindent
Observation~\ref{prop:iv} holds by definition of~$A$. Observations~\ref{prop:v} and~\ref{prop:vi} follow from combining observations~\ref{prop:ii} and~\ref{prop:iii}.

Finally, using properties~\ref{prop:iv}--\ref{prop:vi}, it is straightforward to prove that there exists a truth assignment of~${\cal C}$ that makes exactly one of the three variables in each clause true if and only if~$G$ is near-bipartite; see also~\cite{YY06}.

\medskip
\noindent
{\bf Our Adjustments.} 
To prove that for $k\geq 3$ the problem of deciding whether a graph of maximum degree~$2k-2$ has a $k$-degenerate decomposition is \NP-complete, we reduce from {\sc $1$-in-$k$-SAT} with positive literals only.
This problem is readily seen to be \NP-complete via a reduction from {\sc $1$-in-$3$-SAT} with positive literals only.
We adjust the construction of Yang and Yuan in the following way:
\begin{enumerate}
\item Change the graph $H(a,b,c,d,e)$ to the graph $H(a,B,D)$, where~$B$ and~$D$ are sets of $k-1$ vertices.
To obtain this graph, replace the (adjacent) vertices~$b$ and~$c$ by a clique on the set~$B$ and replace the (non-adjacent) vertices~$d$ and~$e$ by an independent set on the set~$D$.
Make every vertex in~$B$ adjacent to every vertex in~$D$ (just as every vertex in~$\{b,c\}$ was adjacent to every vertex in~$\{d,e\}$).\\[-8pt]
\item Construct the graphs~$H^i$ using the following modified construction:
\begin{itemize}
\item Set $H^0=H(a,\{\tensor*[^1]{b}{_1^0},\ldots,\tensor*[^{(k-1)}]{b}{_1^0}\},\{\tensor*[^1]{d}{_1^0},\ldots,\tensor*[^{(k-1)}]{d}{_1^0}\})$.\\[-6pt]
\item For $i \geq 1$, construct~$H^{i-1}$ and~$(k-1)^i$ copies of~$H$, that is, for $1 \leq j \leq 2^i$, construct
$H(a,\{\tensor*[^1]{b}{_j^i},\ldots,\tensor*[^{(k-1)}]{b}{_j^i}\},\{\tensor*[^1]{d}{_j^i},\ldots,\tensor*[^{(k-1)}]{d}{_j^i}\})$.
Next, for $1 \leq j \leq (k-1)^{i-1}$, and $1 \leq \ell \leq k-1$, identify the vertices $\tensor*[^\ell]{d}{_j^{i-1}}$ and~$a_{(k-1)(j-1)+\ell}^i$.\\[-6pt]
\end{itemize}
\item Take~$H^q$ to be the vertex gadgets, where $(k-1)^{q-1}<m\leq (k-1)^q$.\\[-8pt]
\item Change each clause gadget from a triangle into a complete graph on~$k$ vertices (corresponding to the~$k$ variables of that clause).\\[-8pt]
\item Replace the two edges going from each variable vertex in each clause gadget to the corresponding~$d$ and~$e$ vertices in the corresponding variable gadget by $k-1$ edges going to the $k-1$ vertices of the independent set~$D$ replacing these~$d$ and~$e$ vertices.
\end{enumerate}

The above adjustments in the construction of Yang and Yuan~\cite{YY06} increase the maximum degree from~$4$ to $2(k-1)$.
This can be seen as follows.
First, each variable vertex in the clause gadget has $k-1$ neighbours in the corresponding variable gadget in addition to its $k-1$ neighbours in the clause gadget.
Moreover, in each variable gadget, all but one of the vertices that correspond to the vertex~$a$ in the graph~$H$ have degree $2(k-1)$, as such vertices were adjacent to vertices~$b$ and~$c$ of two different subgraphs of type~$H$ and have now been made adjacent to every vertex of two cliques~$B$ of size $k-1$.
Every vertex in a clique~$B$ that replaced a pair of vertices~$b$ and~$c$ also has $2k-2$ neighbours, namely $k-2$ neighbours in that clique, one neighbour corresponding to~$a$ and $k-1$ neighbours in the independent set~$D$ that replaced~$d$ and~$e$.
Finally, each vertex in an independent set that replaced the vertices~$d$ and~$e$, but that has not been identified with a vertex~$a$ has $k-1$ neighbours in the clique~$B$ that replaced~$b$ and~$c$ and at most one neighbour in a corresponding clause gadget, if one exists.

Note that for $k=3$ we obtain the construction of Yang and Yuan~\cite{YY06}.
By exactly the same arguments as in their proof for $k=3$, we can prove the following result that extends the result in~\cite{YY06}.

\begin{theorem}
For every $k\geq 3$, the problem of deciding whether a graph of maximum degree $2k-2$ has a $k$-degenerate decomposition is \NP-complete.
\end{theorem}

\section{Application: Reconfigurations of Vertex Colourings}\label{s-reconfig}

Our interest in finding $k$-degenerate decompositions stems from an open problem in the area of graph reconfigurations.
For a graph~$G$ and an integer~$k\geq 1$, the $k$-colouring reconfiguration graph~$R_k(G)$ has vertex set consisting of all possible $k$-colourings of~$G$ and two vertices of~$R_k(G)$ are adjacent if and only if the corresponding $k$-colourings differ on exactly one vertex.
The following problem has been the subject of much study; see e.g.~\cite{BLPP14,BC09,BMNR14,Ce07,FJP16,JKKPP16}: 

\medskip
\noindent
\emph{Given a graph~$G$ on~$n$ vertices and two $k$-colourings~$\alpha$ and~$\beta$ of~$G$, find a path (if one exists) in~$R_k(G)$ between~$\alpha$ and~$\beta$.}

\medskip
\noindent
In this section, we are concerned with determining, for every pair $k,\Delta$, the complexity of this problem on graphs~$G$ with maximum degree~$\Delta$.
In Section~\ref{s-proof} we will prove the following result using Theorem~\ref{thm:k2-deg}.

\begin{proposition}\label{prop1}
Let~$G$ be a connected graph on~$n$ vertices with maximum degree $\Delta \geq 3$.
Then it is possible to find a path in~$R_{\Delta+1}(G)$ (if one exists) between any two given $(\Delta+1)$-colourings in~$O(n^2)$ time.
\end{proposition}

In~\cite[Theorem~6]{FJP16}, three of the authors of the current paper proved Proposition~\ref{prop1} in all cases except where the input graph~$G$ is $\Delta$-regular.
Their argument used the fact that if~$G$ has maximum degree~$\Delta$, but is not $\Delta$-regular, then~$G$ is $(\Delta-1)$-degenerate.
In this case it is possible to translate a structural result of Mih\'ok~\cite{Mi01} (also proved by Wood~\cite{Wo05}) into an $O(n^2)$-time algorithm, as shown in~\cite{FJP16}.
However, this does not work if~$G$ is $\Delta$-regular.

The following theorem and proof demonstrate that Proposition~\ref{prop1} indeed fills the gap left by past work on this problem.

\begin{theorem}\label{thm:reconfig}
Let $\Delta\geq 0$ be a (fixed) integer.
Let~$G$ be a connected graph on~$n$ vertices with maximum degree~$\Delta$.
The problem of finding a path (if one exists) between two $k$-colourings~$\alpha$ and~$\beta$ in~$R_k(G)$ is
\begin{itemize}
\item $O(n)$-time solvable if $1\leq k\leq 3$;\\[-11pt]
\item $O(n^2)$-time solvable if $k\geq 4$ and $0\leq \Delta\leq k-1$;\\[-11pt]
\item \emph{PSPACE}-hard if $k\geq 4$ and $\Delta\geq k$.
\end{itemize}
\end{theorem}

\begin{proof}
In~\cite{JKKPP16}, the problem was shown to be solvable in $O(n+m)$ time on (general) graphs with~$n$ vertices and~$m$ edges for $k\leq 3$.
An~$O(n^2)$ time algorithm for the case where $k\geq 4$, $0\leq \Delta\leq k-2$ was presented in~\cite{Ce07}.
In~\cite{BC09}, PSPACE-hardness for $k\geq 4$, $\Delta\geq k$ was proved.
This leaves only the case where $k\geq 4$ and $\Delta=k-1$, which follows from Proposition~\ref{prop1}.\qedllncs
\end{proof}

\subsection{The Proof of Proposition~\ref{prop1}}\label{s-proof}

It was already known~\cite[Theorem~2]{FJP16} that for every connected graph on~$n$ vertices with maximum degree $\Delta \geq 3$, there is a path of length~$O(n^2)$ between any two given $(\Delta+1)$-colourings in~$R_{\Delta+1}(G)$ unless one or both of the $(\Delta+1)$-colourings is an isolated vertex in~$R_{\Delta+1}(G)$.
To prove Proposition~\ref{prop1}, we have to show how to {\em find} such paths between colourings in~$R_{\Delta+1}(G)$ in~$O(n^2)$ time.
Apart from using Theorem~\ref{thm:k2-deg}, this requires us to replace several structural lemmas of~\cite{FJP16} by their algorithmic counterparts.
As we have also managed to simplify some of the arguments from~\cite{FJP16}, we present a self-contained proof of Proposition~\ref{prop1} in this section.

At several places in the proof of Proposition~\ref{prop1} we seek to show that from some given colouring~$\alpha$ of a graph~$G$, we can find a path in~$R_{\Delta+1}(G)$ to another colouring with some specified property.
Rather than explicitly referring to paths in~$R_{\Delta+1}(G)$, we think, equivalently, in terms of \emph{recolouring} vertices of~$G$ one by one in order to turn~$\alpha$ into the colouring we require.

We now define a number of terms that we will use to describe vertices of~$G$ with respect to a $(\Delta+1)$-colouring~$\alpha$.
A vertex~$v$ is \emph{locked} by~$\alpha$ if~$\Delta$ distinct colours appear on its neighbours; note that in this case every neighbour of~$v$ has a unique colour.
A vertex that is not locked is \emph{free}.
Clearly a vertex can be recoloured only if it is free.
A vertex~$v$ is \emph{superfree} if there is a colour $c \neq \Delta+1$ such that neither~$v$ nor any of its neighbours is coloured~$c$, that is, $c\neq \alpha(u)$ for any~$u$ in the closed neighbourhood $N[v]=\{u\; |\; uv\in E\}\cup \{v\}$ of~$v$.
A vertex~$v$ can be recoloured with a colour other than $\Delta+1$ if and only if~$v$ is superfree.
For any two colours $j,k \in \{1,\ldots,\Delta+1\}$, a \emph{$(j,k)$-component} is a connected component in the subgraph of~$G$ induced by the vertices coloured~$j$ or~$k$.
As we continually recolour, these terms should be assumed to be used with respect to the current colouring unless specified otherwise.
We let~$L_\alpha$ denote the set of vertices~$u$ with colour $\alpha(u)=\Delta+1$.
We say that we {\em compact}~$\alpha$ if we determine a path in~$R_{\Delta+1}(G)$ from~$\alpha$ to a $(\Delta+1)$-colouring~$\alpha^*$ of~$G$ with $|L_{\alpha^*}|< |L_{\alpha}|$.
We say that~$\alpha$ has the {\em lock-property} if, for every $u\in L_\alpha$ and every~$v\in N[u]$, we have that~$v$ is locked.

The following lemma is crucial.

\begin{lemma}\label{lem:p-1}
Let~$G$ be a connected graph on~$n$ vertices with maximum degree $\Delta \geq 3$ that is not isomorphic to~$K_{\Delta+1}$.
Let~$\alpha$ be a $(\Delta+1)$-colouring of~$G$ that is not an isolated vertex in~$R_{\Delta+1}(G)$ such that $L_\alpha\neq \emptyset$.
Then it is possible to compact~$\alpha$ in~$O(n)$ time.
\end{lemma}

\begin{proof}
\setcounter{ctrclaim}{0}
We note that~$G$ has~$O(n)$ edges as~$\Delta$ is a fixed constant.
Let $L_\alpha=\{u_1,\ldots,u_p\}$ for some $p\geq 1$.
We first prove a series of claims, which will enable us to deal with several special cases.

\clm{\label{clm2:1}Given a free vertex~$v$ in~$N[u_i]$ for some $i\in \{1,\ldots,p\}$, we can compact~$\alpha$ in~$O(1)$ time.}
We prove Claim~\ref{clm2:1} as follows.
First suppose that~$u_i$ is free.
Let $c \neq \Delta+1$ be a colour not used on~$N[u_i]$.
We can determine~$c$ in~$O(1)$ time, as~$\Delta$ is a constant.
Then we can recolour~$u_i$ with colour~$c$.
Now assume that~$u_i$ is not free, that is, $u_i$ is locked.
Then~$v$ is a neighbour of~$u_i$.
We determine a colour~$c'$ not used on~$N[v]$ in~$O(1)$ time such that $c' \neq \Delta+1$.
We recolour~$v$ with colour~$c'$, and now~$u_i$ is free and can be recoloured with $\alpha(v) \neq \Delta+1$.
Hence, we have compacted~$\alpha$ in~$O(1)$ time.
This proves Claim~\ref{clm2:1}.

\clm{\label{clm2:2}Let~$j$ and~$k$ be distinct colours in $\{1,\ldots,\Delta\}$.
If a $(j,k)$-component~$D$ is such that no vertex coloured~$j$ in~$D$ has a neighbour in~$L_\alpha$, then in~$O(|V(D)|)$ time we can recolour~$G$ from~$\alpha$ to the $(\Delta+1)$-colouring~$\alpha'$ with $\alpha'(v)=\alpha(v)$ if $v\notin V(D)$ and $\alpha'(v)=j+k-\alpha(v)$ if $v \in D$ (that is, the colours on~$D$ are \emph{swapped}).}
We prove Claim~\ref{clm2:2} as follows.
We recolour all vertices of~$D$ that have colour~$j$ with colour~$\Delta+1$.
This yields a new $(\Delta+1)$-colouring, as none of these vertices has a neighbour in~$L_\alpha$.
We then recolour all vertices of~$D$ that have colour~$k$ with the colour~$j$.
Again this is a valid operation as, by the choice of~$D$, all their neighbours that were given colour~$j$ by~$\alpha$ are now coloured $\Delta+1$.
We finally recolour all vertices of~$D$ that were given colour~$j$ by~$\alpha$ (and so have $\Delta+1$ in the current colouring) with the colour~$k$.
This yields the $(\Delta+1)$-colouring~$\alpha'$.
Moreover, the running time of doing this is linear in the size of~$D$.
This proves Claim~\ref{clm2:2}.

\medskip
\noindent
Suppose that~$\alpha$ has the lock-property.
For every $i \in \{1,\ldots,p\}$ and every $j \in \{1,\ldots,\Delta\}$, we denote the unique neighbour of~$u_i$ coloured~$j$ by~$v_{i,j}$.
For every $i \in \{1,\ldots,p\}$ and every $j, k \in \{1,\ldots,\Delta\}$ ($j \neq k$), the graph~$G_i^{j,k}$ is the $(j,k)$-component containing~$v_{i,j}$.
We note that~$G_i^{j,k}$ and~$G_i^{k,j}$ may or may not be equal.
These definitions are with respect to~$\alpha$ unless stated otherwise.
We will write, for example, $\tensor*[^\alpha]{G}{*_i^{j,k}}$ if we need to specify the colouring.

\clm{\label{clm2:3}Suppose that~$\alpha$ has the lock-property.
Let~$i$ be in $\{1,\ldots,p\}$ and let~$j$ and~$k$ be two distinct colours in $\{1,\ldots,\Delta\}$.
If~$G_i^{j,k}$ is \emph{not} a path where each end-vertex is locked and no vertex is superfree, then we can compact~$\alpha$ in~$O(n)$ time.}
We start with two observations.
If a vertex $v\in V(G_i^{j,k})$ has degree at least~$2$ in~$G_i^{j,k}$, then~$v$ has two neighbours with the same colour, so it is free.
By the lock-property, for all $\ell \in \{1,\ldots,p\}$, every vertex in~$N[u_\ell]$ is locked.
Hence~$v$ is not adjacent to a vertex in~$L_\alpha$.
If a vertex $v\in V(G_i^{j,k})$ has degree at least~$3$ in~$G_i^{j,k}$, then~$v$ is also superfree.

As~$v_{i,j}$ is in~$N[u_i]$, by the lock-property, we find that~$v_{i,j}$ is locked.
Hence~$v_{i,j}$ has degree~$1$ in~$G_i^{j,k}$.
We now perform a breadth-first search in~$G_i^{j,k}$ starting from~$v_{i,j}$.
This takes~$O(n)$ time, as~$\Delta$ is a constant.
We stop if we find a superfree vertex~$w$ or else if we have visited all vertices of~$G_i^{j,k}$.

First suppose that we found a superfree vertex~$w$.
As all vertices closer to~$v_{i,j}$ than~$w$ in~$G_i^{j,k}$ are not superfree, they must have degree~$2$ and form a path~$P$ from~$v_{i,j}$ to~$w$.
As the internal vertices of~$P$ have degree~$2$ in~$G_i^{j,k}$, they are free and so have no neighbour in~$L_\alpha$ by the lock-property.
Since~$w$ is superfree, there is some colour $c \neq \Delta+1$ with which we can recolour~$w$.
As colours~$j$ and~$k$ appear on~$N[w]$, we find that $c \not\in \{j,k\}$.
After we recolour~$w$ with~$c$, we note that~$G_i^{j,k}$ (defined with respect to this new colouring) is a $(j,k)$-component where no vertex coloured~$k$ has a neighbour in~$L_\alpha$.
We apply Claim~\ref{clm2:2} and note that now~$u_i$ has no neighbour coloured~$j$.
We therefore recolour~$u_i$ with~$j$ and have compacted~$\alpha$ in~$O(n)$ time.

Now suppose that we found that no vertex in~$G_i^{j,k}$ is superfree.
Then no vertex of~$G_i^{j,k}$ has degree more than~$2$.
Hence~$G_i^{j,k}$ is a path or cycle.
Since~$v_{i,j}$ has degree~$1$, we find that~$G_i^{j,k}$ must be a path.
Let~$z$ be the end-vertex of~$G_i^{j,k}$ other than~$v_{i,j}$.
As no vertex of the path~$G_i^{j,k}$ is superfree and~$v_{i,j}$ is locked, $z$ must be free by the assumption of the claim.
Hence, every vertex of~$G_i^{j,k}$ apart from~$v_{i,j}$ is free, which means that no vertex of~$G_i^{j,k}$ with colour~$k$ has a neighbour in~$L_\alpha$ by the lock-property.
We apply Claim~\ref{clm2:2}.
Afterwards we can recolour~$u_i$ with colour~$j$ to compact~$\alpha$ in~$O(n)$ time.
This proves Claim~\ref{clm2:3}.

\clm{\label{clm2:4}Suppose that~$\alpha$ has the lock-property.
Let~$i_1$ and~$i_2$ be in $\{1,\ldots,p\}$ (we allow the case $i_1=i_2$) and let $j_1$, $j_2$, $k_1$ and~$k_2$ be in $\{1,\ldots,\Delta\}$, $j_1 \neq k_1$, $j_2 \neq k_2$.
If~$G_{i_1}^{j_1,k_1}$ and~$G_{i_2}^{j_2,k_2}$ are two distinct paths, each with locked end-vertices and no superfree vertices, then~$G_{i_1}^{j_1,k_1}$ and~$G_{i_2}^{j_2,k_2}$ do not intersect on a free vertex.}
We prove Claim~\ref{clm2:4} as follows.
For contradiction, suppose that~$G_{i_1}^{j_1,k_1}$ and~$G_{i_2}^{j_2,k_2}$ intersect on a free vertex~$w$.
Then~$w$ is an internal vertex on each of~$G_{i_1}^{j_1,k_1}$ and~$G_{i_2}^{j_2,k_2}$.
Moreover, the two neighbours of~$w$ in~$G_{i_1}^{j_1,k_1}$ are disjoint from the two neighbours of~$w$ in~$G_{i_2}^{j_2,k_2}$, otherwise $\{j_1,k_1\}=\{j_2,k_2\}$, in which case $G_{i_1}^{j_1,k_1}=G_{i_2}^{j_2,k_2}$, a contradiction.
Thus~$w$ has four neighbours that use only two colours between them.
So~$w$ has at most $\Delta-2$ colours in its neighbourhood.
This means that~$w$ is superfree, a contradiction.
This proves Claim~\ref{clm2:4}.

\clm{\label{clm2:5}Suppose that~$\alpha$ has the lock-property.
Let~$i$ be in $\{1,\ldots,p\}$ and let~$j$ and~$k$ be in $\{1,\ldots,\Delta\}$ with $j \neq k$.
If~$G_i^{j,k}$ and~$G_i^{k,j}$ are two distinct paths, each with locked end-vertices and no superfree vertices, and moreover, $G_i^{j,k}$ does not contain exactly two vertices, then it is possible to compact~$\alpha$ in~$O(n)$ time.}
We prove Claim~\ref{clm2:5} as follows.
Since $v_{i,j}\in N[u_i]$ is locked by the lock-property, $v_{i,j}$ has exactly one neighbour~$s$ coloured~$k$.
This means that~$G_i^{j,k}$ has at least two vertices, and thus at least three vertices by the assumption of the claim.
Hence~$s$ has another neighbour with colour~$j$.
This means that~$s$ is free (note that~$s$ is not superfree by the assumption of the claim).

Let~$t$ be a neighbour of~$s$ not in~$G_i^{j,k}$.
Recall that $\Delta \geq 3$, so~$t$ exists.
Let $c=\alpha(t)$, and since~$s$ is free, $c \neq \Delta+1$ by the lock-property.
As~$t$ does not belong to~$G_i^{j,k}$, we find that $c\notin \{j,k\}$.

First suppose that~$t$ is locked.
Then, by definition, $t$ has a neighbour with colour~$\Delta+1$.
Let $u_h \in L_\alpha$ be this neighbour for some $h \in \{1,\ldots,p\}$.
If $h=i$, then~$G_i^{j,k}$ and~$G_i^{c,k}$ intersect on the free vertex~$s$, contradicting Claim~\ref{clm2:4}.
Hence, $h \neq i$.
Then we recolour~$s$ with~$\Delta+1$.
This is possible, as~$s$ is free and any vertices adjacent to vertices in~$L_\alpha$ are locked by the lock-property.
We then recolour~$v_{i,j}$ with colour~$k$.
As~$t$ is locked by~$\alpha$, we find that~$\alpha$ has coloured exactly one neighbour of~$t$ with colour~$j$ and exactly one neighbour (namely~$s$) with colour~$k$.
If~$t$ is adjacent to~$v_{i,j}$ then this is the neighbour of~$t$ that is coloured~$j$ by~$\alpha$; in this case we recolour~$t$ with colour~$j$.
If~$t$ is non-adjacent to~$v_{i,j}$ then we recolour~$t$ with colour~$k$.
Thus we can recolour~$t$ either with colour~$j$ or with colour~$k$.
Now both~$u_h$ and~$u_i$ are free and can be recoloured.
In this way in~$O(n)$ time we have reduced the total number of vertices coloured $\Delta+1$ by $2-1=1$, that is, we have compacted~$\alpha$.

Now suppose that~$t$ is free.
As the end-vertex of~$G_i^{j,k}$ other than~$v_{i,j}$ is locked by assumption of the claim, it has a neighbour $u_\ell\in L_\alpha$.
If $\ell = i$, then $G_i^{j,k}=G_i^{k,j}$, which is not possible by the assumption of the claim.
Hence we have $\ell \neq i$.
We let~$z_1$ be the neighbour of~$t$ on~$G_i^{j,k}$ closest to~$u_i$, and we let~$z_2$ be the neighbour of~$t$ on~$G_i^{j,k}$ closest to~$u_\ell$.
Note that $z_1=z_2=s$ is possible.
We now recolour~$t$ with colour $\Delta+1$.
This is possible, as~$t$ is free and thus has no neighbour in~$L_\alpha$ by the lock-property.
As~$z_2$ is not superfree, $z_2$ is either locked or free.
In the first case~$\alpha$ colours each neighbour of~$z_2$ with a unique colour.
In the second case~$z_2$ is an inner vertex of~$G_i^{j,k}$ and, as~$z_2$ is not superfree, $\alpha$ colours exactly two neighbours of~$z_2$ with the same colour, which is either~$j$ or~$k$.
Hence, $t$ is the only neighbour of~$z_2$ which~$\alpha$ colours with colour~$c$.
Similarly, $t$ is the only neighbour of~$z_1$ with colour~$c$.
Hence, after recolouring~$t$ with colour $\Delta+1$, we can recolour~$z_1$ and~$z_2$ with~$c$.
All the internal vertices in~$G_i^{j,k}$ are free, so by the lock-property, they have no neighbours in~$L_\alpha$.
Next we apply Claim~\ref{clm2:2} to the two obtained subpaths of $G_i^{j,k} \setminus \{z_1, z_2\}$ containing neighbours of~$u_i$ and~$u_\ell$, respectively.
Afterwards, both~$u_i$ and~$u_\ell$ can be recoloured.
In this way in~$O(n)$ time we have reduced the number of vertices coloured~$\Delta+1$ by $2-1=1$, that is, we have compacted~$\alpha$.
Hence we have proven Claim~\ref{clm2:5}.

\medskip
\noindent
{\bf Our Algorithm.}
If~$\alpha$ does not have the lock-property then we are done by Claim~\ref{clm2:1}.
We may therefore assume that~$\alpha$ has the lock-property.
As~$\alpha$ is not an isolated vertex in~$R_{\Delta+1}(G)$, $G$ has at least one free vertex~$x$.
Let~$P$ be a shortest path in~$G$ from~$x$ to a vertex in~$L_\alpha$, say to~$u_1$.
We may assume that~$x$ is chosen such that~$x$ is the free vertex on~$P$ closest to~$u_1$.
Then every internal vertex of~$P$ is locked and coloured with a colour in $\{1,\ldots,\Delta\}$.
Moreover, the only vertex of~$P$ with a neighbour in~$L_\alpha$ is the neighbour of~$u_1$.
By definition, any locked vertex not in~$L_\alpha$ has a neighbour with colour~$\Delta+1$, so it is adjacent to a vertex in~$L_\alpha$.
Hence~$P$ has at most two edges.
As~$x$ is free, but by the lock-property every vertex in~$N[u_1]$ is locked, it follows that~$u_1$ and~$x$ are not adjacent.
Therefore, $P$ contains exactly two edges.
Without loss of generality, we may assume that the middle vertex is~$v_{1,1}$, which has colour~$1$ by definition, and that~$x$ is coloured~$2$.
In particular note that in this case $x\in V(G_1^{1,2})$.

We may assume that both~$G_1^{1,2}$ and~$G_1^{2,1}$ are paths whose end-vertices are locked and that do not contain any superfree vertices, otherwise we apply Claim~\ref{clm2:3} and are done.
As~$G_1^{1,2}$ is a path whose end-vertices are locked and $x\in V(G_1^{1,2})$ is free, $G_1^{1,2}$ has at least three vertices.
Hence we may assume that $G_1^{1,2}=G_1^{2,1}$, otherwise we may apply Claim~\ref{clm2:5} and are done.

Let~$H^{2,3}$ be the maximal $(2,3)$-component of~$G$ containing~$x$.
As~$x$ is not superfree and has two neighbours coloured~$1$, $x$ has only one neighbour coloured~$3$.
Hence~$x$ has degree~$1$ in~$H^{2,3}$.
If~$H^{2,3}$ is not a path, then let~$w$ be the vertex with degree at least~$3$ in~$H^{2,3}$ that is closest to~$x$.
As~$w$ has three neighbours coloured alike, $w$ is superfree.
This means we can recolour~$w$ with a colour $c \neq \Delta+1$.
Note that $c \notin \{2,3\}$.

This recolouring of~$w$ may have altered the graph~$G_1^{1,2}$ only if $c=1$.
If~$G_1^{1,2}$ is no longer a path, then we apply Claim~\ref{clm2:3} and are done in~$O(n)$ time.
Hence suppose that~$G_1^{1,2}$ is (still) a path and note that the recolouring of~$w$ also changed~$H^{2,3}$ into a path (if~$H^{2,3}$ was not already a path).
For simplicity, we still call the current colouring~$\alpha$.

The internal vertices of the path~$H^{2,3}$ each have two neighbours coloured alike.
Hence, they are not locked and therefore all but at most one vertex of~$H^{2,3}$ is free.
By the lock-property, every neighbour of a vertex in~$L_\alpha$ is locked.
Consequently, no internal vertex of~$H^{2,3}$ has a neighbour in~$L_\alpha$.
Let~$x'$ be the end-vertex of~$H^{2,3}$ other than~$x$.
Then~$\alpha(x')$ is either~$2$ or~$3$.
If $\alpha(x')=2$, then we apply Claim~\ref{clm2:2} with $j=3$ and $k=2$.
If $\alpha(x')=3$, then we apply Claim~\ref{clm2:2} with $j=2$ and $k=3$ (as~$x$ is free and thus has no neighbour in~$L_\alpha$ either).
Let~$\beta$ be the resulting colouring.

We now proceed by applying a similar procedure to~$\beta$ to the one we applied to~$\alpha$.
If~$\beta$ does not have the lock-property, then we are done by Claim~\ref{clm2:1}, so we may assume that~$\beta$ does have the lock-property.
Recall that~$v_{1,2}$ and~$v_{1,3}$ are locked by~$\alpha$.
Since at most one vertex in~$H^{2,3}$ is locked by~$\alpha$, it follows that at most one vertex in $\{v_{1,2},v_{1,3}\}$ is in~$H^{2,3}$.
If~$v_{1,2}$ is in~$H^{2,3}$, but~$v_{1,3}$ is not then $\beta(v_{1,2})=\beta(v_{1,3})=3$, so~$u_1$ is not locked by~$\beta$, contradicting the lock-property.
Therefore~$v_{1,2}$ is not in~$H^{2,3}$, and similarly~$v_{1,3}$ is not in~$H^{2,3}$, so $\beta(v_{1,2})=2$ and $\beta(v_{1,3})=3$.
We may assume that~$\tensor*[^\beta]{G}{*_1^{1,2}}$ and~$\tensor*[^\beta]{G}{*_1^{2,1}}$ are paths whose end-vertices are locked and that do not contain any superfree vertices, otherwise we apply Claim~\ref{clm2:3} and are done.
Note that $\tensor*[^\alpha]{G}{*_1^{1,2}}=\tensor*[^\alpha]{G}{*_1^{2,1}}$ consists of a path which contains the vertices~$v_{1,1}$, $x$, and~$v_{1,2}$ in that order and~$\alpha$ colours these vertices $1$, $2$ and~$2$, respectively.
Therefore~$v_{1,2}$ must have a neighbour~$v_{1,2}'$ in~$\tensor*[^\alpha]{G}{*_1^{2,1}}$ that~$\alpha$ colours with colour~$1$, and this must be an internal vertex of~$\tensor*[^\alpha]{G}{*_1^{2,1}}$, so by the lock-property, it cannot be adjacent to a vertex in~$L_\alpha$.
Since $L_\alpha=L_\beta$ it follows that~$v_{1,2}'$ has no neighbours in~$L_\beta$, and so it must be free in~$\beta$.
As~$\tensor*[^\beta]{G}{*_1^{2,1}}$ is a path whose end-vertices are locked and $v_{1,2}\in V(\tensor*[^\beta]{G}{*_1^{2,1}})$ is free, $\tensor*[^\beta]{G}{*_1^{2,1}}$ has at least three vertices.
Thus we may assume that $\tensor*[^\beta]{G}{*_1^{1,2}}=\tensor*[^\beta]{G}{*_1^{2,1}}$, otherwise we apply Claim~\ref{clm2:5} and are done.

As~$v_{1,1}$ is locked by~$\beta$, we find that~$v_{1,1}$ has a neighbour~$z$ with $\beta(z)=2$.
Hence, $z$ belongs to~$\tensor*[^\beta]{G}{*_1^{1,2}}$.
Recall that $\alpha(x)=2$.
As~$v_{1,1}$ is locked by~$\alpha$, we find that~$v_{1,1}$ has no other neighbour that got colour~$2$ in the colouring~$\alpha$, by the lock-property.
Note that $\beta(x)=3$.
Hence, in order for~$v_{1,1}$ to have a neighbour coloured~$2$ by~$\beta$, it must be the case that~$z$ belongs to~$H^{2,3}$ and thus $\alpha(z)=3$.
As~$v_{1,2}$ is not in~$H^{2,3}$, we find that $z\neq v_{1,2}$.
Then~$z$ is an internal vertex of $\tensor*[^\beta]{G}{*_1^{1,2}}=\tensor*[^\beta]{G}{*_1^{2,1}}$.
Thus~$z$ has two neighbours coloured~$1$.
If~$z$ is an internal vertex of~$H^{2,3}$, then~$\beta$ colours two other neighbours of~$z$ with colour~$3$.
This implies that~$z$ is superfree, contradicting the fact that~$\tensor*[^\beta]{G}{*_1^{1,2}}$ has no superfree vertices.
Thus~$z$ is the end-vertex of~$H^{2,3}$ other than~$x$, that is, $z=x'$.

We now let~$y$ be the first vertex of~$\tensor*[^\beta]{G}{*_1^{1,2}}$ that is not in~$\tensor*[^\alpha]{G}{*_1^{1,2}}$ when traversing~$\tensor*[^\beta]{G}{*_1^{1,2}}$ from~$v_{1,2}$.
Note that such a vertex~$y$ exists, as~$z$ belongs to~$\tensor*[^\beta]{G}{*_1^{1,2}}$ but not to~$\tensor*[^\alpha]{G}{*_1^{1,2}}$ (as $\alpha(z)=3$).
Thus $\alpha(y) \neq \beta(y)$, so~$y$ belongs to~$H^{2,3}$.
The end-vertices of~$\tensor*[^\beta]{G}{*_1^{1,2}}$ are~$v_{1,1}$ and~$v_{1,2}$, neither of which belong to~$H^{2,3}$.
Hence, $y$ is an inner vertex of~$\tensor*[^\beta]{G}{*_1^{1,2}}$.
If~$y$ is an internal vertex of both~$H^{2,3}$ and~$\tensor*[^\beta]{G}{*_1^{1,2}}$, then~$y$ would be superfree, contradicting the fact that~$\tensor*[^\beta]{G}{*_1^{1,2}}$ has no superfree vertices.
This means that~$y$ is an end-vertex of~$H^{2,3}$.
As $\beta(x)=3$ while $\beta(y)=1$ or $\beta(y)=2$, we find that $y \neq x$.
It follows that $y=z=x'$, and so $\beta(y)=2$.

Consider the vertex~$z'$ on~$\tensor*[^\beta]{G}{*_1^{1,2}}$ reached immediately before $y=z=x'$ when traversing~$\tensor*[^\beta]{G}{*_1^{1,2}}$ from~$v_{1,2}$.
Then $\alpha(z')=\beta(z')=1$.
As $\alpha(v_{1,2})=\alpha(y)=2$, we find that~$z'$ is an inner vertex of~$\tensor*[^\alpha]{G}{*_1^{1,2}}$.
It follows that~$z'$ is a free vertex.
Note that~$z$ is adjacent to~$v_{1,1}$ and~$z'$ and that $\alpha(z)=3$ and $\alpha(v_{1,1})=\alpha(z')=1$.
Therefore~$z'$ is a vertex of~$\tensor*[^\alpha]{G}{*_1^{1,3}}$ and so~$\tensor*[^\alpha]{G}{*_1^{1,3}}$ and~$\tensor*[^\alpha]{G}{*_1^{1,2}}$ intersect on the free vertex~$z'$.
By Claim~\ref{clm2:4}, it follows that~$\tensor*[^\alpha]{G}{*_1^{1,3}}$ not a path whose end-vertices are locked and that does not contain any superfree vertices, and so we are done by applying Claim~\ref{clm2:3}.
This completes the description of the algorithm.

\medskip
\noindent
The correctness of our algorithm follows directly from its description.
Hence it remains to discuss its runtime.

\medskip
\noindent
{\bf Runtime analysis.}
We first compute the set~$L_\alpha$ in~$O(n)$ time, as~$\Delta$ is a constant.
We then apply Claim~\ref{clm2:1} on each~$u_i$.
As this takes~$O(1)$ time per vertex, obtaining the lock-property takes~$O(n)$ in total.
As~$\Delta$ is a constant, for a given pair~$(i,j)$, the vertex~$v_{i,j}$ can be found in~$O(1)$ time.
Moreover, for a given triple~$(i,j,k)$, we can compute~$G_i^{j,k}$ in~$O(n)$ time.
As~$\Delta$ is a constant, we can find a free vertex~$x$ in~$O(n)$ time.

We can find the path~$P$ to a vertex in~$L_\alpha$, which we assumed was~$u_1$, by using a breadth-first search starting from~$x$.
As~$\Delta$ is a constant, this takes~$O(n)$ time.
We may also assume that~$x$ is the free vertex on~$P$ closest to~$u_1$ on this path, as otherwise we can replace~$x$ by some other free vertex of~$P$ in~$O(n)$ time.

We check in~$O(n)$ time whether~$G_1^{1,2}$ is a path whose end-vertices are locked and that does not contain any superfree vertices.
We check the same in~$O(n)$ time for~$G_1^{2,1}$.
If we find that for at least one of these graphs this is not the case, then our algorithm applies Claim~\ref{clm2:3} (either with $i=j=1$, $k=2$ or with $i=k=1$, $j=2$), and we are done in~$O(n)$ time.
So suppose this is the case for both~$G_1^{1,2}$ and~$G_1^{2,1}$.
Then we check in~$O(n)$ time whether $G_1^{1,2}=G_1^{2,1}$.
If not, then our algorithm applies Claim~\ref{clm2:5} (with $i=j=1$, $k=2$), and we are done in~$O(n)$ time.
So suppose that $G_1^{1,2}=G_1^{2,1}$.

Recolouring the vertex~$w$ in $H^{2,3}$ (if it exists) and applying Claim~\ref{clm2:2} (either with $j=3$, $k=2$ or with $j=2$, $k=3$) takes~$O(n)$ time.
So far we have used~$O(n)$ time.
Hence, it takes~$O(n)$ time to do the similar procedure for the colouring~$\beta$ obtained after applying Claim~\ref{clm2:2}.
We find the vertices~$z$ and~$z'$ in~$O(1)$ time.
Afterwards we apply Claim~\ref{clm2:3}, which takes~$O(n)$ time.
So we used~$O(n)$ time in total.
This completes the proof of the lemma.\qedllncs
\end{proof}

We are now ready to prove Proposition~\ref{prop1} for graphs~$G$ with maximum degree~$\Delta\geq 3$ by following the arguments from~\cite{FJP16} without the requirement that~$G$ is $(\Delta-1)$-degenerate (that is, we allow~$G$ to be $\Delta$-regular).
So the proof is similar to the proof in~\cite{FJP16} for $(\Delta-1)$-degenerate graphs except that we use Theorem~\ref{thm:k2-deg} instead of its algorithmic counterpart for $(\Delta-1)$-degenerate graphs.
To show this we give a self-contained proof.

\medskip
\noindent
{\bf Proposition~\ref{prop1} [restated].}
{\em Let~$G$ be a connected graph on~$n$ vertices with maximum degree $\Delta \geq 3$.
Then it is possible to find a path in~$R_{\Delta+1}(G)$ (if one exists) between any two given $(\Delta+1)$-colourings in~$O(n^2)$ time.}

\begin{proof}
We use induction on~$\Delta$.
If $\Delta \in \{1,2\}$, then the statement is trivially true.
Let $\Delta\geq 3$ and assume that we have an~$O(n^2)$ time algorithm for connected graphs on~$n$ vertices with maximum degree~$\Delta-1$.
Applying Lemma~\ref{lem:p-1}~$O(n)$ times, in~$O(n^2)$ time we can find a path from~$\alpha$ to some $\Delta$-colouring~$\gamma_1$ and a path from~$\beta$ to some $\Delta$-colouring~$\gamma_2$.
By Theorem~\ref{thm:k2-deg} we can find in~$O(n^2)$ time a partition $\{S_1, S_2\}$ of~$V(G)$ such that~$S_1$ is an independent set and~$S_2$ induces a $(\Delta-2)$-degenerate graph, which we denote by~$H$.
We modify the pair $(S_1,S_2)$ in~$O(n^2)$ time by moving vertices from~$S_2$ to~$S_1$ until~$S_1$ is a maximal independent set.

Let~$\gamma_1^H$ and~$\gamma_2^H$ be the $\Delta$-colourings of~$H$ that are the restrictions of~$\gamma_1$ and~$\gamma_2$, respectively, to~$S_2$.
Let~$\gamma_1'$ and~$\gamma_2'$ be the $(\Delta+1)$-colourings obtained from~$\gamma_1$ and~$\gamma_2$, respectively, by recolouring every vertex in~$S_1$ with the colour $\Delta+1$.
As~$S_1$ is maximal, $H$ has maximum degree at most $\Delta-1$.
We apply the induction hypothesis to find in~$O(n^2)$ time a path between the two $\Delta$-colourings~$\gamma_1^H$ and~$\gamma_2^H$ in~$R_{\Delta}(H)$ (note that neither is an isolated vertex of~$R_{\Delta}(H)$ since~$H$ is $(\Delta-2)$-degenerate).
Note that this immediately translates into a path between~$\gamma_1'$ and~$\gamma_2'$ in~$R_{\Delta+1}(G)$.
Hence we obtain a path between~$\alpha$ and~$\beta$ in~$R_{\Delta+1}(G)$.
This completes the proof.\qedllncs
\end{proof}

\section{Future Work}\label{s-con}

In this section we pose two open problems.
We have proven that for every integer $k\geq 3$, the problem of finding a $k$-degenerate decomposition is polynomial-time solvable on graphs of maximum degree~$k$ and \NP-hard for graphs of maximum degree~$2k-\nobreak 2$ (by generalizing the hardness proof of~\cite{YY06} for $k=3$).
This brings us to our first open problem.

\begin{open}
Determine, for every integer $k\geq 4$, the complexity of finding a $k$-degenerate decomposition for graphs of maximum degree~$k+\nobreak 1$.
\end{open}

Our second open problem is related to Theorem~\ref{t-ma07}.
Recall that this theorem states that for every three integers $k\geq 3$ and $p,q\geq 0$ with $p+q=k-2$, the vertex set of every connected graph of maximum degree~$k$ that is not isomorphic to~$K_{k+1}$ can be partitioned into two sets~$A$ and~$B$, where~$A$ induces a $p$-degenerate subgraph of maximum size and~$B$ induces a $q$-degenerate subgraph.
Our algorithms in Sections~\ref{s-3} and~\ref{s-maxdegree} form an algorithmic version of Theorem~\ref{t-special}, which is a special case of Theorem~\ref{t-ma07} in which $p=1$ and $q=k-1$.

\begin{open}
Does there exist an algorithmic version of Theorem~\ref{t-ma07} similar to our algorithmic version of Theorem~\ref{t-special}?
\end{open}

\bibliography{mybib}

\begin{thebibliography}{10}

\bibitem{AGSS16}
A.~Agrawal, S.~Gupta, S.~Saurabh, and R.~Sharma.
\newblock Improved algorithms and combinatorial bounds for independent feedback
  vertex set.
\newblock {\em Proc. IPEC 2016, LIPIcs}, 63:2:1--2:14, 2017.

\bibitem{BW14}
B.~Baetz and D.~R. Wood.
\newblock Brooks' vertex-colouring theorem in linear time.
\newblock {\em CoRR}, abs/1401.8023, 2014.

\bibitem{BDFJP17b}
M.~Bonamy, K.~K. Dabrowski, C.~Feghali, M.~Johnson, and D.~Paulusma.
\newblock Independent feedback vertex set for ${P_5}$-free graphs.
\newblock Manuscript, 2017.

\bibitem{BDFJP17c}
M.~Bonamy, K.~K. Dabrowski, C.~Feghali, M.~Johnson, and D.~Paulusma.
\newblock Independent feedback vertex sets for graphs of bounded diameter.
\newblock Manuscript, 2017.

\bibitem{BDFJP17-conf}
M.~Bonamy, K.~K. Dabrowski, C.~Feghali, M.~Johnson, and D.~Paulusma.
\newblock Recognizing graphs close to bipartite graphs.
\newblock {\em Proc. MFCS 2017, LIPIcs}, 83:70:1--70:14, 2017.

\bibitem{BLPP14}
M.~Bonamy, M.~Johnson, I.~Lignos, V.~Patel, and D.~Paulusma.
\newblock Reconfiguration graphs for vertex colourings of chordal and chordal
  bipartite graphs.
\newblock {\em Journal of Combinatorial Optimization}, 27(1):132--143, 2014.

\bibitem{BC09}
P.~S. Bonsma and L.~Cereceda.
\newblock Finding paths between graph colourings: {PSPACE-completeness} and
  superpolynomial distances.
\newblock {\em Theoretical Computer Science}, 410(50):5215--5226, 2009.

\bibitem{BMNR14}
P.~S. Bonsma, A.~E. Mouawad, N.~Nishimura, and V.~Raman.
\newblock The complexity of bounded length graph recoloring and {CSP}
  reconfiguration.
\newblock {\em Proc. IPEC 2014, LNCS}, 8894:110--121, 2014.

\bibitem{Bo76}
O.~V. Borodin.
\newblock On decomposition of graphs into degenerate subgraphs.
\newblock {\em Diskretny\u\i \ Analiz}, 28:3--11, 1976.
\newblock (in Russian).

\bibitem{BG01}
O.~V. Borodin and A.~N. Glebov.
\newblock On the partition of a planar graph of girth~$5$ into an empty and an
  acyclic subgraph.
\newblock {\em Diskretny\u\i \ Analiz i Issledovanie Operatsi\u\i . Seriya 1},
  8:34--53, 2001.
\newblock (in Russian).

\bibitem{BKT00}
O.~V. Borodin, A.~V. Kostochka, and B.~Toft.
\newblock Variable degeneracy: extensions of {Brooks'} and {Gallai's} theorems.
\newblock {\em Discrete Mathematics}, 214(1--3):101--112, 2000.

\bibitem{BBKNP13}
A.~Brandst{\"a}dt, S.~Brito, S.~Klein, L.~T. Nogueira, and F.~Protti.
\newblock Cycle transversals in perfect graphs and cographs.
\newblock {\em Theoretical Computer Science}, 469:15--23, 2013.

\bibitem{BHLL05}
A.~Brandst{\"a}dt, P.~L. Hammer, V.~B. Le, and V.~V. Lozin.
\newblock Bisplit graphs.
\newblock {\em Discrete Mathematics}, 299(1--3):11--32, 2005.

\bibitem{BLS98}
A.~Brandst{\"a}dt, V.~B. Le, and T.~Szymczak.
\newblock The complexity of some problems related to graph 3-colorability.
\newblock {\em Discrete Applied Mathematics}, 89(1--3):59--73, 1998.

\bibitem{Br41}
R.~L. Brooks.
\newblock On colouring the nodes of a network.
\newblock {\em Mathematical Proceedings of the Cambridge Philosophical
  Society}, 37(2):194--197, 1941.

\bibitem{CC96}
L.~Cai and D.~G. Corneil.
\newblock A generalization of perfect graphs -- $i$-perfect graphs.
\newblock {\em Journal of Graph Theory}, 23(1):87--103, 1996.

\bibitem{Catlin79}
P.~A. Catlin.
\newblock Brooks' graph-coloring theorem and the independence number.
\newblock {\em Journal of Combinatorial Theory, Series B}, 27(1):42--48, 1979.

\bibitem{CL95}
P.~A. Catlin and H.-J. Lai.
\newblock Vertex arboricity and maximum degree.
\newblock {\em Discrete Mathematics}, 141(1--3):37--46, 1995.

\bibitem{Ce07}
L.~Cereceda.
\newblock {\em Mixing graph colourings}.
\newblock PhD thesis, London School of Economics, 2007.

\bibitem{CK69}
G.~Chartrand and H.~V. Kronk.
\newblock The point-arboricity of planar graphs.
\newblock {\em Journal of the London Mathematical Society}, s1-44(1):612--616,
  1969.

\bibitem{DLS15}
K.~K. Dabrowski, V.~V. Lozin, and J.~Stacho.
\newblock {Stable-$\Pi$} partitions of graphs.
\newblock {\em Discrete Applied Mathematics}, 182:104--114, 2015.

\bibitem{DMP16}
F.~Dross, M.~Montassier, and A.~Pinlou.
\newblock Partitioning a triangle-free planar graph into a forest and a forest
  of bounded degree.
\newblock {\em European Journal of Combinatorics}, (in press).

\bibitem{FHKM03}
T.~Feder, P.~Hell, S.~Klein, and R.~Motwani.
\newblock List partitions.
\newblock {\em {SIAM} Journal on Discrete Mathematics}, 16(3):449--478, 2003.

\bibitem{FJP16}
C.~Feghali, M.~Johnson, and D.~Paulusma.
\newblock A reconfigurations analogue of {Brooks'} {Theorem} and its
  consequences.
\newblock {\em Journal of Graph Theory}, 83(4):340--358, 2016.

\bibitem{GJS76}
M.~R. Garey, D.~S. Johnson, and L.~J. Stockmeyer.
\newblock Some simplified {NP-complete} graph problems.
\newblock {\em Theoretical Computer Science}, 1(3):237--267, 1976.

\bibitem{GSSH93}
D.~L. Grinstead, P.~J. Slater, N.~A. Sherwani, and N.~D. Holmes.
\newblock Efficient edge domination problems in graphs.
\newblock {\em Information Processing Letters}, 48(5):221--228, 1993.

\bibitem{GLS84}
M.~Gr\"otschel, L.~Lov\'asz, and A.~Schrijver.
\newblock Polynomial algorithms for perfect graphs.
\newblock {\em Annals of Discrete Mathematics}, 21:325--356, 1984.

\bibitem{HSW90}
S.~L. Hakimi, E.~F. Schmeichel, and J.~Weinstein.
\newblock Partitioning planar graphs into independent sets and forests.
\newblock {\em Congressus Numerantium}, 78:109--118, 1990.

\bibitem{HL00}
C.~T. Ho\`ang and V.~B. Le.
\newblock On {$P_4$-transversals} of perfect graphs.
\newblock {\em Discrete Mathematics}, 216(1--3):195--210, 2000.

\bibitem{JKKPP16}
M.~Johnson, D.~Kratsch, S.~Kratsch, V.~Patel, and D.~Paulusma.
\newblock Finding shortest paths between graph colourings.
\newblock {\em Algorithmica}, 75(2):295--321, 2016.

\bibitem{KT09}
K.~Kawarabayashi and C.~Thomassen.
\newblock Decomposing a planar graph of girth 5 into an independent set and a
  forest.
\newblock {\em Journal of Combinatorial Theory, Series B}, 99(4):674--684,
  2009.

\bibitem{Ko78}
A.~V. Kostochka.
\newblock {\em Upper bounds of chromatic functions of graphs}.
\newblock PhD thesis, University of Novosibirsk, 1978.
\newblock (in Russian).

\bibitem{KS97}
J.~Kratochv\'{\i}l and I.~Schiermeyer.
\newblock On the computational complexity of
  {$(\mathcal{O},\mathcal{P})$-partition} problems.
\newblock {\em Discussiones Mathematicae Graph Theory}, 17(2):253--258, 1997.

\bibitem{Lo72}
L.~Lov\'asz.
\newblock A characterization of perfect graphs.
\newblock {\em Journal of Combinatorial Theory, Series B}, 13(2):95--98, 1972.

\bibitem{Lo73}
L.~Lov\'asz.
\newblock Coverings and coloring of hypergraphs.
\newblock {\em Congressus Numerantium}, VIII:3--12, 1973.

\bibitem{Lo05}
V.~V. Lozin.
\newblock Between 2- and 3-colorability.
\newblock {\em Information Processing Letters}, 94(4):179--182, 2005.

\bibitem{MP95}
N.~V.~R. Mahadev and U.~N. Peled.
\newblock {\em Threshold graphs and related topics}, volume~56 of {\em Annals
  of Discrete Mathematics}.
\newblock North-Holland, Amsterdam, 1995.

\bibitem{Ma07}
M.~Matamala.
\newblock Vertex partitions and maximum degenerate subgraphs.
\newblock {\em Journal of Graph Theory}, 55(3):227--232, 2007.

\bibitem{MY15}
C.~McDiarmid and N.~Yolov.
\newblock Recognition of unipolar and generalised split graphs.
\newblock {\em Algorithms}, 8(1):46--59, 2015.

\bibitem{Mi01}
P.~Mih\'ok.
\newblock Minimal reducible bounds for the class of $k$-degenerate graphs.
\newblock {\em Discrete Mathematics}, 236(1--3):273--279, 2001.

\bibitem{MPRS12}
N.~Misra, G.~Philip, V.~Raman, and S.~Saurabh.
\newblock On parameterized independent feedback vertex set.
\newblock {\em Theoretical Computer Science}, 461:65--75, 2012.

\bibitem{Oc05}
P.~Ochem.
\newblock Negative results on acyclic improper colorings.
\newblock {\em Proc. EuroComb 2005 Discrete Mathematics \& Theoretical Computer
  Science}, AE:357--362, 2005.

\bibitem{TIZ15}
Y.~Tamura, T.~Ito, and X.~Zhou.
\newblock Algorithms for the independent feedback vertex set problem.
\newblock {\em {IEICE} Transactions on Fundamentals of Electronics,
  Communications and Computer Sciences}, E98-A(6):1179--1188, 2015.

\bibitem{Th95}
C.~Thomassen.
\newblock Decomposing a planar graph into degenerate graphs.
\newblock {\em Journal of Combinatorial Theory, Series B}, 65(2):305--314,
  1995.

\bibitem{Th01}
C.~Thomassen.
\newblock Decomposing a planar graph into an independent set and a 3-degenerate
  graph.
\newblock {\em Journal of Combinatorial Theory, Series B}, 83(2):262--271,
  2001.

\bibitem{Wo05}
D.~R. Wood.
\newblock Acyclic, star and oriented colourings of graph subdivisions.
\newblock {\em Discrete Mathematics \& Theoretical Computer Science}, 7:37--50,
  2005.

\bibitem{WYZ96}
Y.~Wu, J.~Yuan, and Y.~Zhao.
\newblock Partition a graph into two induced forests.
\newblock {\em Journal of Mathematical Study}, 29:1--6, 1996.

\bibitem{YY06}
A.~Yang and J.~Yuan.
\newblock Partition the vertices of a graph into one independent set and one
  acyclic set.
\newblock {\em Discrete Mathematics}, 306(12):1207--1216, 2006.

\end{thebibliography}
\end{document}